\documentclass[aps,preprintnumbers,nofootinbib,onecolumn]{revtex4}
\usepackage{graphics}
\usepackage{amsfonts,amssymb,amsmath}            
\usepackage{ifthen}
\usepackage{amsthm}
\usepackage{dsfont}
\usepackage{float}
\usepackage{MnSymbol}
\usepackage[caption=false]{subfig}
\usepackage{tikz}
\usepackage[hyperfootnotes=false]{hyperref}
\usepackage{xcolor}
\definecolor{mylinkcolor}{rgb}{0,0,0.8} 
\hypersetup{unicode=true,%
  bookmarksnumbered=false,bookmarksopen=false,bookmarksopenlevel=1, %
  breaklinks=true,pdfborder={0 0 0},colorlinks=true}%
\hypersetup{%
  anchorcolor=mylinkcolor,citecolor=mylinkcolor, %
  filecolor=mylinkcolor,linkcolor=mylinkcolor, %
  menucolor=mylinkcolor,runcolor=mylinkcolor, %
  urlcolor=mylinkcolor}

\usepackage[nodayofweek]{datetime}
\newcommand{\comment}[1]{}
\newcommand{\ket}[1]{| #1 \rangle}
\newcommand{\bra}[1]{\langle #1 |}

\newcommand{\proj}[1]{|#1\rangle\!\langle#1|}
\newcommand{\id}{\openone}

\newcommand{\cG}{\mathcal{G}}
\newcommand{\cH}{\mathcal{H}}

\newcommand{\cS}{\mathcal{S}}

\newcommand{\rC}{\mathrm{C}}
\newcommand{\rG}{\mathrm{GPT}}
\newcommand{\rQ}{\mathrm{Q}}

\newcommand{\ot}{\otimes}

\DeclareMathOperator{\Tr}{Tr}

\theoremstyle{plain}
\newtheorem{theorem}{Theorem}
\newtheorem{conjecture}{Conjecture}
\newtheorem{lemma}[theorem]{Lemma}
\newtheorem{corollary}[theorem]{Corollary}

\theoremstyle{definition}
\newtheorem{definition}{Definition}
\newtheorem{remark}{Remark}

\begin{document}

\title{Analysing causal structures using Tsallis entropies}
\author{V. Vilasini}
\email{vv577@york.ac.uk}
\affiliation{Department of Mathematics, University of York, Heslington, York YO10 5DD.}
\author{Roger Colbeck}
\email{roger.colbeck@york.ac.uk}
\affiliation{Department of Mathematics, University of York, Heslington, York YO10 5DD.}
\date{\today}

\begin{abstract}
  Understanding cause-effect relationships is a crucial part of the
  scientific process.  As Bell's theorem shows, within a given causal
  structure, classical and quantum physics impose different
  constraints on the correlations that are realisable, a fundamental
  feature that has technological applications. However, in general it
  is difficult to distinguish the set of classical and quantum
  correlations within a causal structure.  Here we investigate a
  method to do this based on using entropy vectors for Tsallis
  entropies. We derive constraints on the Tsallis entropies that are
  implied by (conditional) independence between classical random
  variables and apply these to causal structures.  We find that the
  number of independent constraints needed to characterise the causal
  structure is prohibitively high such that the computations required for
  the standard entropy vector method cannot be employed even for small
  causal structures.  Instead, without solving the whole problem, we
  find new Tsallis entropic constraints for the triangle causal
  structure by generalising known Shannon constraints. Our results reveal
  new mathematical properties of classical and quantum Tsallis
  entropies and highlight difficulties of using Tsallis entropies
  for analysing causal structures. 
\end{abstract}

\maketitle

\section{Introduction} \label{sec:introduction} Cause-effect
relationships between physical systems constrain the correlations that
can arise between them. The study of causality allows us to explain
observed correlations between different variables in terms of
unobserved systems that cause these variables to become
correlated. This has found applications in diverse fields of research
such as medical testing, socio-economic surveys and physics. The
foundational interest in causal structures stems from the fact that
the theory that describes the unobserved systems affects the set of
possible correlations over the observed variables. Bell
inequalities~\cite{Bell} are constraints on the observed correlations
in a classical causal structure (Figure~\ref{fig: Bell}) and can be
violated in quantum and generalised probabilistic theories (GPTs).
The possibility of such violations leads to applications in
device-independent
cryptography~\cite{Ekert91,MayersYao,BHK,RogerThesis,CK,Pironio2009}.

In the bipartite Bell causal structure (Figure~\ref{fig: Bell}), the
set of all joint conditional distributions $P_{XY|AB}$ over the
observed nodes $X,Y,A,B$ that can arise when $\Lambda$ is classical is
relatively well understood.  For fixed input and output sizes, it
forms a convex polytope and hence membership can be checked using a
linear program (although the size of the linear program scales
exponentially with the number of inputs and the problem is
NP-complete~\cite{Pitowski89}).  Because of this, the complete set of
Bell inequalities characterizing these polytopes are unknown for
$|X|,|Y|>3$ or $|A|,|B|>5$~\cite{Masanes2002, Bancal2010, Cope2019}.

In causal structures with more unobserved common causes (such as the
triangle causal structure of Figure~\ref{fig: triangle}), the set of
compatible correlations is not well understood. The inflation
technique~\cite{wolfe2016} can in principle certify whether or not a
given distribution belongs to the classical marginal entropy
cone\footnote{The set of possible entropy vectors over the observed
  nodes of the classical causal structure.} of a causal
structure~\cite{Navascues2017}. However, the method does not tell us
how to construct a suitable inflation of the causal structure in order
to achieve this, or how large this inflation needs to be. Thus, in
general, using the inflation technique becomes intractable in
practice.  The difficulty of solving the general problem in part stems
from its non-convexity. One approach to overcoming this is
to analyse the problem in entropy space~\cite{Yeung97}. This has
proven to be useful in a number of cases (see, e.g.,~\cite{Fritz13,
  Chaves13}, or~\cite{Weilenmann20170483} for a detailed review),
since the problem is convex in entropy space and the entropic
inequalities characterising the relevant sets are independent of the
number of measurement outcomes. However, it was shown
in~\cite{Weilenmann16} that the entropy vector method with Shannon
entropies cannot detect the classical-quantum gap for line-like causal
structures.\footnote{Note that this result holds for the entropic
  characterisation without post-selection. Using the post-selection
  technique (see e.g.,~\cite{Weilenmann20170483} for an explanation),
  one can derive quantum-violatable Shannon entropic inequalities even
  for line-like causal structures~\cite{BraunsteinCaves88} however
  this technique is not generalizable to causal structures that have
  no parentless observed nodes, such as in Figure~\ref{fig: triangle}.} Further, even though new Shannon
entropic inequalities have been derived using this method, no quantum
violation of these have been found for a range of causal structures
where non-classical correlations are known to
exist~\cite{Weilenmann2018, Chaves2014}.  Due to these limitations of
Shannon entropies, it is natural to ask whether other entropic
quantities could do better.

Here we consider Tsallis entropies in the entropy vector method for
analysing causal structures. One motivation for considering such
entropies for the task is that they are a family with an additional
(real) parameter.  The set of entropies for all possible values of
this parameter conveys more information about the underlying
probability distribution than a single member of the family and hence
the ability to vary a parameter may give advantages for analysing causal structures.  Tsallis entropies appear to be a good candidate since they satisfy monotonicity, submodularity
and the chain rule which are desirable properties for their use in the
entropy vector method.\footnote{Other examples of more
  general entropy measures such as the R\'enyi entropy~\cite{renyi} do
  not satisfy one or more of these properties, making it more difficult to get entropic constraints on them using the entropy vector method.} Tsallis entropies
have been considered in the context of causal structures
before~\cite{Wajs15} where they were shown to give an advantage over
Shannon entropy in detecting the non-classicality of certain states in
the Bell scenario if one also post-selects on the values of observed
parentless nodes\footnote{Note that non-classicality cannot be
  detected entropically in the Bell causal structure (Figure~\ref{fig:
    Bell}) without post-selection \cite{Weilenmann16}.}.   Here we
consider a systematic treatment that can be applied to an arbitrary
causal structure in the absence of post-selection. (Note that use of post-selection is not possible in causal structures with no observed parentless nodes such as the Triangle of Figure~\ref{fig: triangle}.)

In Section~\ref{sec: tsalcaus}, we derive the constraints
on the classical Tsallis entropies that are implied by a given causal
structure and in Appendix~\ref{appendix: quantum}, we generalise this
result to quantum Tsallis entropies for certain cases. In Section
\ref{ssec: tsalentvec}, we use these constraints in the entropy vector
method with Tsallis entropies but find that the computational
procedure becomes too time consuming even for simple causal structures
such as the bipartite Bell scenario. Despite this limitation, we
derive new Tsallis entropic inequalities for the triangle causal
structure in Section~\ref{sec: newineq}, using known Shannon entropic
inequalities of~\cite{Chaves2014} and our Tsallis constraints of
Section~\ref{sec: tsalcaus}. In Section~\ref{sec: causaldiscussion},
we discuss the reasons for the computational difficulty of this
method, the drawbacks of using Tsallis entropies for analysing causal
structure and identify potential future directions.

\begin{figure}[t!]
	\centering
	\subfloat[]{\includegraphics[scale=1.0]{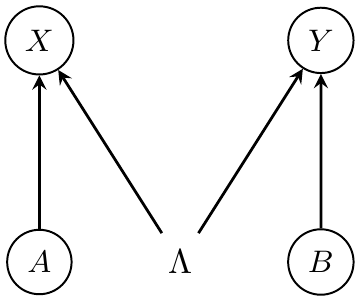}
	\label{fig: Bell}}\qquad\qquad\qquad\qquad\qquad\qquad
\subfloat[]{\includegraphics[scale=1.0]{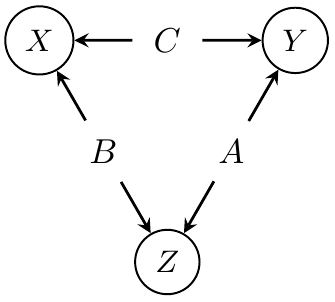}
	\label{fig: triangle}}
\caption{\textbf{Some causal structures:} Observed nodes are circled and the uncircled ones correspond to unobserved nodes. (a) The bipartite Bell causal structure. The nodes $A$ and $B$ represent the random variables corresponding to Alice's and Bob's choice of input while $X$ and $Y$ represent the random variables corresponding to their outputs. $\Lambda$ here is the only potentially unobserved node and is the common cause of $X$ and $Y$. (b) The triangle causal structure. Here, the three observed nodes $X$, $Y$ and $Z$ have unobserved, pairwise common causes $A$, $B$ and $C$, but no joint common cause.}
\end{figure}

\section{Shannon entropy and the entropy vector method}
\label{sec:entvec}
Given a random variable $X$ distributed according to the discrete probability
distribution\footnote{We will only be considering random variables defined on a finite set in this paper.} $p_X$, the Shannon entropy of $X$ is given by $H(X)=-\sum_x p_X(x)\ln p_X(x)$.\footnote{Note that it is common to
  take logarithms in base 2 and measure entropy in \emph{bits}; here
  we use base $e$ corresponding to measuring entropy in \emph{nats}.}
Given two random variables $X$ and $Y$, distributed according to
$P_{XY}$ the conditional Shannon entropy is defined by
$H(X|Y)=-\sum_{x,y}p_{XY}(xy)\ln \frac{p_{XY}(xy)}{p_Y(y)}$ and the
Shannon mutual information by $I(X:Y)=H(X)-H(X|Y)$.  For three random
variables $X$, $Y$ and $Z$, we can also define the mutual information
between $X$ and $Y$ conditioned on $Z$, $I(X:Y|Z)=H(X|Z)-H(X|YZ)$.

We will sometimes use the shorthands $p_x=p_X(x)=p(X=x)$ and
$p_{x|y}:=p_{X|Y}(X=x|Y=y)$ etc.\ for probability
distributions.

We next provide a short overview of the entropy vector method that
suffices for the purposes of this paper. For a more detailed overview
of the method, see~\cite{Weilenmann20170483}. Consider a joint
distribution $p_{X_1,\ldots,X_n}$ over $n$ random variables
$X_1,X_2,\ldots,X_n$. With each such distribution, we associate a vector
with $2^n-1$ components, each of which correspond to the entropy of an
element of the powerset of $\{X_1,X_2,\ldots,X_n\}$ (excluding the empty
set). This defines the \emph{entropy vector} of $p_{X_1,\ldots,X_n}$.
Note that this vector encodes the conditional entropies and mutual
informations via the relations $H(X|Y)=H(XY)-H(Y)$,
$I(X:Y)=H(X)+H(Y)-H(XY)$ and $I(X:Y|Z)=H(XZ)+H(YZ)-H(XYZ)-H(Z)$.  We
use $\mathbf{H}$ to denote
the map that takes a probability distribution over $n$ variables to
its entropy vector (with $2^n-1$ components) and
$\Gamma_n^*$ to denote the set of all
vectors that are entropy vectors of a probability distribution
$p_{X_1,\ldots,X_n}$, i.e., $\Gamma_n^*=\{v\in \mathbb{R}^{2^n-1}: \exists p_{X_1,\ldots,X_n} \text{
  s.t. } v=\mathbf{H}(p_{X_1,\ldots,X_n})\}$. The
closure of $\Gamma_n^*$, denoted by $\overline{\Gamma_n^*}$ is known
to be a convex set for any $n$ \cite{Zhang1997ANC}.
 
\subsection{The Shannon cone}
 
Valid entropy vectors necessarily satisfy certain constraints. These
include positivity of the entropies, monotonicity (i.e.,
$H(R)\leq H(RS)$) and submodularity (also known as
strong-subadditivity; $H(RT)+H(ST)\geq H(RST)+H(T)$). Monotonicity and
submodularity are equivalent to the positivity of the conditional
entropy $H(S|R)$ and the conditional mutual information $I(R:S|T)$
respectively and hold for any three disjoint subsets $R$, $S$ and $T$
of $\{X_1,\ldots,X_n\}$. This set of linear constraints are together
known as the \emph{Shannon constraints} and the set of vectors
$u \in \mathbb{R}^{2^n-1}$ obeying all the Shannon constraints form
the convex cone known as the \emph{Shannon cone}, $\Gamma_n$. Other
than positivity (which, following standard practice, we include
implicitly), there are a total of $n+n(n-1)2^{n-3}$ independent Shannon
constraints for $n$ variables \cite{Yeung97}. By definition, the
Shannon cone is an outer approximation to $\overline{\Gamma_n^*}$
i.e., $\overline{\Gamma_n^*} \subseteq \Gamma_n$.\footnote{For
  $n\leq3$, the cones coincide, but for $n\geq 4$ they do not
  \cite{Zhang1997ANC}.} Hence all entropy vectors derived from a
probability distribution $p_{X_1,\ldots,X_n}$ obey the Shannon
constraints but not all vectors $u \in \mathbb{R}^{2^n-1}$ obeying the
Shannon constraints are such that $\mathbf{H}(p_{X_1,\ldots,X_n})=u$
for some joint distribution $p_{X_1,\ldots,X_n}$. 

In the next subsection we discuss how causal structures give
additional entropic constraints.

\subsection{Entropy vectors and causal structure} \label{ssec:
  entcaus} A causal structure can be represented as a Directed Acyclic
Graph (DAG) over several nodes, some of which are labelled observed
and some unobserved. Each observed node corresponds to a classical
random variable\footnote{These may represent inputs or outputs of an experiment.}, while for each unobserved node there is an
associated system whose nature depends on the theory being
considered. A causal structure is called \emph{classical} (denoted
$\mathcal{G}^\rC$), \emph{quantum} (denoted $\mathcal{G}^\rQ$) or
\emph{GPT} (denoted $\mathcal{G}^\rG$) depending on the nature of the
unobserved nodes. In the following, we briefly review the framework of
classical causal models~\cite{Pearl2000}.

A distribution $p_{X_1,\ldots,X_n}$ over $n$ random variables $\{X_1,\ldots,X_n\}$ is said to be \emph{compatible} with a classical causal structure $\mathcal{G}^\rC$ (with these variables as nodes) if it satisfies the \emph{causal Markov condition} i.e., the joint distribution decomposes as
\begin{equation}
\label{eq: markov}
    p_{X_1,\ldots,X_n}=\prod\limits_{i=1}^{n} p_{X_i|X_i^{\downarrow_1}},
\end{equation}
where $X_i^{\downarrow_1}$ denotes the set of all parent nodes of the
node $X_i$ in the DAG $\mathcal{G}^\rC$. The Markov condition of
Equation~(\ref{eq: markov}) is equivalent to the conditional
independence of $X_i$ from its non-descendants, denoted
$X_i^{\nuparrow}$ given its parents $X_i^{\downarrow_1}$ in
$\mathcal{G}$ i.e., $\forall i\in \{1,\ldots,n\}$,
$p_{X_i
  X_i^{\nuparrow}|X_i^{\downarrow_1}}=p_{X_i|X_i^{\downarrow_1}}p_{X_i^{\nuparrow}|X_i^{\downarrow_1}}$~\cite{Pearl2000}. All
other conditional independences between different subsets of nodes are
implied by these $n$ constraints and can be derived from these
constraints and standard probability calculus based on Bayes' rule.  The
concept of \emph{d-separation} developed by Geiger~\cite{Geiger1987}
and Verma and Pearl~\cite{Verma2013} provides a method to read off
implied conditional independence relations from the graph.  In other
words, for arbitrary disjoint subsets $X$, $Y$ and $Z$ of the nodes,
it can be used to determine whether $X$ and $Y$ are conditionally
independent given $Z$.

\begin{definition}[Blocked paths]
Let $\mathcal{G}$ be a DAG in which $X$ and $Y\neq X$ are nodes and $Z$ be a
set of nodes not containing $X$ or $Y$.  A path from $X$ to $Y$ is
said to be \emph{blocked} by $Z$ if it contains either $A\rightarrow W\rightarrow B$
with $W\in Z$, $A\leftarrow W\rightarrow B$
with $W\in Z$ or $A\rightarrow W\leftarrow B$ such that neither $W$ nor any descendant of $W$ belongs to $Z$.
\end{definition}

\begin{definition}[d-separation]
Let $\mathcal{G}$ be a DAG in which $X$, $Y$ and $Z$ are disjoint
sets of nodes.  $X$ and $Y$ are \emph{d-separated} by $Z$ in
$\mathcal{G}$ if every path from a variable in $X$ to a variable in
$Y$ is \emph{blocked} by $Z$.
\end{definition}

The importance of d-separation is that, given a causal structure
$\cG$, $X$ and $Y$ are d-separated by $Z$ in $\cG$ if and only if
$I(X:Y|Z)=0$ for all distributions compatible with
$\cG$~\cite{Pearl2000}.\footnote{Note that $I(X:Y|Z)=0$ is equivalent to
  $P_{XY|Z}=P_{X|Z}P_{Y|Z}$.}  The complete set of d-separation
conditions give all the conditional independence relations implied by
the DAG.  In the case of Shannon entropy for a DAG with $n$ nodes
these are all implied by the $n$ constraints
\begin{equation}
\label{eq: shancausmain}
  I(X_i:X_i^{\nuparrow}|X_i^{\downarrow_1})=0 \quad \forall i\in\{1,\ldots,n\}. 
\end{equation}
In other words, a distribution over $n$ variables satisfies
Equation~\eqref{eq: markov} if and only if it satisfies
Equation~\eqref{eq: shancausmain}.

Since we wish to contrast classical and quantum versions of causal
structures we also define the latter. For the purpose of this work, it
is sufficient to do so for causal structures with at most two
generations and in which the first generation can be either observed
classical random variables or unobserved quantum nodes, while those of
the second generation are only observed classical variables (in
Appendix~\ref{appendix: quantum} we also look at a case in which the
second generation can be quantum). Each edge emanating from an
unobserved node has an associated Hilbert space labelled by the parent
and the child. For example an edge from an unobserved node $X$ to an
observed node $Y$ has the associated Hilbert space
$\mathcal{H}_{X_Y}$. Each unobserved quantum node corresponds to a
density operator in the tensor product of the Hilbert space
corresponding to all the edges emanating from that node. For each
observed node, there is a POVM that acts on the tensor product of the
Hilbert spaces associated with the edges that meet at that node. The
set of distributions over observed nodes compatible with the quantum
causal structure $\mathcal{G}^Q$ are those that can be obtained by
performing the specified POVMs (possibly specified by classical input
nodes in the first generation) on the relevant quantum states and
using the Born rule.  For instance, a distribution $P_{ABXY}$ is
compatible with the quantum analogue of Figure~\ref{fig: Bell} if
there exists a quantum state
$\rho\in\cH_{\Lambda_X}\ot\cH_{\Lambda_Y}$ and POVMs $\{E^a_x\}_x$ and
$\{F^b_y\}_y$ acting on $\cH_{\Lambda_X}$ and $\cH_{\Lambda_Y}$
respectively, such that
$P_{ABXY}(a,b,x,y)=P_A(a)P_B(b)\Tr(\rho(E^a_x\ot F^b_y))$ for all
values of the random variables.

Now, in the case of classical causal structures with unobserved nodes, the
compatibility condition requires that there exists a joint
distribution $p_{X_1,\ldots,X_n}$ over the $n$ variables satisfying the
causal Markov condition and having the correct marginals over the
observed nodes. In quantum and more general theories, the existence of
a joint state over all the nodes is not guaranteed because there may
be sets of systems that do not coexist. (For example, there is no
joint quantum state of a system and the outcome of a measurement on
it.) Because classical information can be copied, such joint
distributions always exist in the classical
case. The entropy vector method aims to exploit this difference to
certify the non-classicality of correlations.

The entropic constraints over all the nodes will in general imply
constraints on the entropy vector over the observed nodes. These can
be obtained by Fourier-Motzkin elimination~\cite{Williams1986}. The
procedure takes the entropy cone over all nodes, that is constrained
by the $n+n(n-1)2^{n-3}$ Shannon constraints and the $n$ causal
constraints (Equation~\eqref{eq: shancausmain}) and projects it onto the
entropy cone of the observed nodes (eliminating all combinations of
entropies involving unobserved nodes). Since non-classical causal
structures do not satisfy the initial assumption of the existence of
the joint distribution/entropies, they may give rise to correlations
that do not satisfy the marginal constraints on the observed nodes
obtained through this procedure. A violation of one of the
inequalities certifies the non-classicality of that causal
structure.

For line-like causal structures (of which the bipartite Bell causal
structure of Figure~\ref{fig: Bell} is an instance), the classical and
quantum Shannon entropy cones coincide and Shannon entropic
inequalities cannot certify the non-classicality of these causal
structures even though they support non-classical
correlations~\cite{Weilenmann16}. Further, in other scenarios such as
the triangle which is also known to support non-classical
correlations~\cite{Fritz2012}, known Shannon entropic inequalities
such as those of~\cite{Chaves2014, Weilenmann2018} have no known
quantum violations. The main question of the current work is whether
using Tsallis entropies can provide tighter, quantum violatable
entropic inequalities and avoid these limitations.

\section{Tsallis entropies}
\label{sec: tsallis}
For a classical random variable $X$ distributed according to the discrete probability distribution $p_X$, the order $q$ Tsallis entropy of $X$ for real parameter $q$ is defined as~\cite{Tsallis1988}
\begin{equation}
\label{eq: tsallis1}
 S_q(X)= 
 \begin{cases}
 -\sum_{\{x:p_x>0\}}p_x^q\ln_q p_x & \text{if } q\neq 1\\
 H(X) & \text{if } q=1
 \end{cases}
\end{equation}
where we have used the short-hand $\ln_q p_x=\frac{p_x^{1-q}-1}{1-q}$.
This $q$-logarithm function converges to the natural logarithm
in the limit $q \rightarrow 1$ so that
$\lim\limits_{q\rightarrow 1}S_q(X)=H(X)$ and the function is
continuous in $q$. For brevity, we will henceforth write $\sum_x$
instead of $\sum_{\{x:p_x>0\}}$, keeping it implicit that probability
zero events do not contribute to the sum.\footnote{Note that this
  means the Tsallis entropy for $q<0$ is not robust in the sense that
  small changes in the probability distribution can lead to large
  changes in the Tsallis entropy.}

The conditional Tsallis entropy~\cite{Furuichi04} is defined by
\begin{equation}
    \label{eq: tsalliscond}
    S_q(X|Y):=  \begin{cases}
 -\sum_{x,y} p_{xy}^q \ln_q p_{x|y}& \text{if } q\neq 1\\
 H(X|Y) & \text{if } q=1
 \end{cases}
\end{equation}
and converges to the Shannon conditional entropy $H(X|Y)$ in the limit $q\rightarrow 1$.
Note that there are other ways to define the conditional Tsallis entropy~\cite{ABE2001157} but they do not satisfy the chain rule (Equation~\eqref{eq: chain2}) and hence will not be considered here.

The unconditional and conditional Tsallis mutual informations are defined analogously to the Shannon case
\begin{equation}
\label{eq: tsalmi1}
    \begin{split}
        &I_q(X:Y)=S_q(X)-S_q(X|Y),\\
        &I_q(X:Y|Z)=S_q(X|Z)-S_q(X|YZ).
    \end{split}
\end{equation}

\subsection{Properties of Tsallis entropies}
\label{ssec: tsalprop}
Tsallis entropies satisfy a number of properties that are desirable
for their use in the entropy vector method. For any joint distribution
over the random variables involved the following properties hold. 
\begin{enumerate}
    \item \label{prop: pseudo}\textbf{Pseudo-additivity \cite{Curado1991}:} For two independent random
      variables $X$ and $Y$ i.e., $p_{XY}=p_Xp_Y$, and for all $q$, the Tsallis entropies satisfy
    \begin{equation}
        \label{eq: pseudoadd}
        S_q(XY)=S_q(X)+S_q(Y)+(1-q)S_q(X)S_q(Y).
    \end{equation}
    Note that in the Shannon case ($q=1$), we recover additivity for
    independent random variables.
  \item \label{prop: upper} \textbf{Upper bound \cite{Furuichirelentropy}:} For $q\geq 0$ we have
    $S_q(X)\leq\ln_qd_X$.  For $q>0$ equality is achieved if and only
    if $P_X(x)=1/d_X$ for all $x$ (i.e., if the distribution on $X$ is
    uniform).
    \item \textbf{Monotonicity \cite{Daroczy1970}:} For all $q$,
    \begin{equation}
        \label{eq: mono}
        S_q(X)\leq S_q(XY).
    \end{equation}
    \item \textbf{Strong subadditivity \cite{Furuichi04}:} For $q\geq1$,
    \begin{equation}
    \label{eq: subadd}
           S_q(XYZ)+S_q(Z)\leq S_q(XZ)+S_q(YZ).
    \end{equation}
   
    \item \textbf{Chain rule~\cite{Furuichi04}:} For all $q$,
    \begin{equation}
        \label{eq: chain2}
        S_q(X_1,X_2,\ldots,X_n|Y)=\sum_{i=1}^n S_q(X_i|X_{i-1},\ldots,X_1,Y).
    \end{equation}
\end{enumerate}
   The chain rules $S_q(XY)=S_q(X)+S_q(Y|X)$ and
   $S_q(XY|Z)=S_q(X|Z)+S_q(Y|XZ)$ emerge as particular cases and allow
   the Tsallis mutual informations of Equation~\eqref{eq: tsalmi1} to be
   written as
    \begin{equation}
\label{eq: tsalmi2}
    \begin{split}
        I_q(X:Y)&=S_q(X)+S_q(Y)-S_q(XY),\\
        I_q(X:Y|Z)&=S_q(XZ)+S_q(YZ)-S_q(Z)-S_q(XYZ).
    \end{split}
\end{equation}

Using the chain rule, the monotonicity and strong subadditivity relations (Equations~\eqref{eq: mono} and~\eqref{eq: subadd}) are equivalent to the non-negativity of the unconditional and conditional Tsallis mutual
informations. For $q<1$, strong subadditivity does not hold in
general~\cite{Furuichi04}, hence we often restrict to the case $q\geq
1$ in what follows.

\section{Causal constraints and Tsallis entropy vectors}
\label{sec: tsalcaus}

In Section \ref{ssec: tsalprop}, we discussed some of the general properties of Tsallis entropy that hold irrespective of the underlying causal structure over the variables. The causal structure imposes the causal Markov constraints on the joint probability distribution over the variables involved (Section \ref{ssec: entcaus}) and we wish to translate these probabilistic constraints into Tsallis entropic ones in order to use Tsallis entropies in the entropy vector method for analysing causal structures. 

A first observation is that Tsallis entropy vectors do not in general satisfy the causal constraints (Equation~\eqref{eq: shancausmain}) satisfied by their Shannon counterparts. For a concrete counterexample, consider the simple, three variable causal structure where $Z$ is a common cause of $X$ and $Y$, and where there are no other causal relations. In terms of Shannon entropies, the only causal constraint in this case is $I(X:Y|Z)=0$. Taking $X,Y$ and $Z$ to be binary variables with possible values $0$ and $1$, the distribution $p_{xyz}=1/4$ $\forall x\in X, y\in Y$ if $z=0$ and $p_{xyz}=0$ otherwise, satisfies $p_{xy|z}=p_{x|z}p_{y|z}$ $\forall x\in X, y\in Y$ and $z\in Z$ but has a $q=2$ Tsallis conditional mutual information of $I_2(X:Y|Z)=\frac{1}{4}$. Hence when using Tsallis entropies (and conditional Tsallis entropy as defined in Section~\ref{sec: tsallis}), the causal constraint cannot be simply encoded by $I_q(X:Y|Z)=0$ for $q>1$.

Given this observation, it is natural to ask whether there are constraints for Tsallis entropies implied by the causal Markov condition (Equation~\eqref{eq: markov}). We answer this question with the following Theorems.

\begin{theorem}
\label{theorem: mibound}
If a joint probability distribution $p_{XY}$ over random variables $X$ and $Y$ with alphabet sizes $d_X$ and $d_Y$ factorises as $p_{XY}=p_Xp_Y$, then for all $q\in[0,\infty)$, the Tsallis mutual information $I_q(X:Y)$ is upper bounded by
\begin{equation*}
    I_q(X:Y)\leq f(q,d_X,d_Y)\,,
\end{equation*}
where the function $f(q,d_X,d_Y)$ is given by
\begin{equation*}
    f(q,d_X,d_Y)=\frac{1}{(q-1)}\left(1-\frac{1}{d_X^{q-1}}\right)\left(1-\frac{1}{d_Y^{q-1}}\right)=(q-1)\ln_qd_X\ln_qd_Y.
\end{equation*}
For $q\in(0,\infty)\setminus\{1\}$, the bound is saturated if and
only if $p_{XY}$ is the uniform distribution over $X$ and $Y$.
\end{theorem}
\begin{proof}
The proof follows from the pseudo-additivity of Tsallis entropies
(Property~\ref{prop: pseudo}) and the upper bound (Property~\ref{prop:
  upper}).  Using these, for all $q\geq 0$ and for all product
distributions $p_{XY}=p_Xp_Y$, we have
\begin{equation}
    \label{eq: mibound}
    I_q(X:Y)=S_q(X)+S_q(Y)-S_q(XY)=(q-1)S_q(X)S_q(Y)\leq \frac{\left(1-\frac{1}{d_X^{q-1}}\right)\left(1-\frac{1}{d_Y^{q-1}}\right)}{q-1}=f(q,d_X,d_Y)\,.
\end{equation}
Whenever $q\in(0,\infty)\setminus\{1\}$, the bound is saturated if
and only if $p_{XY}$ is uniform over $X$ and $Y$ since,
for these values of $q$, $S_q(X)$ and $S_q(Y)$ both attain their
maximum values if and only if this is the case.
\end{proof}

\begin{theorem}
\label{theorem: causal}
If a joint probability distribution $p_{XYZ}$ satisfies the conditional independence $p_{XY|Z}=p_{X|Z}p_{Y|Z}$, then for all $q\geq 1$ the Tsallis conditional mutual information $I_q(X:Y|Z)$ is upper bounded by
\begin{equation*}
    I_{q}(X:Y|Z)\leq f(q,d_X,d_Y)\,.
\end{equation*}
 For $q>1$, the bound is saturated only by distributions in which for
 some fixed value $k$ the joint probabilities are given by
 $p_{xyz}=\begin{cases}\frac{1}{d_Xd_Y} \quad \text{if}\quad z=k\\ 0
   \quad \qquad \text{otherwise}\end{cases}$ for all $x$, $y$ and
 $z$\footnote{These distributions have deterministic $Z$ and there is
   one such distribution for each value that $Z$ can take.}.
\end{theorem}
\begin{proof}
Writing out $I_q(X:Y|Z)$ in terms of probabilities we have
\begin{align*}
    I_q(X:Y|Z)&=\frac{1}{q-1}\big[\sum_{xyz}p^q(xyz)+\sum_z p^q(z) -\sum_{xz} p^q(xz)-\sum_{yz}p^q(yz)\big]\\
    &=\frac{1}{q-1}\sum_zp^q(z)\big[\sum_{xy}p^q(xy|z)+1 -\sum_{x} p^q(x|z)-\sum_{y}p^q(y|z)\big]\\
    &=\sum_z p^q(z)I_q(X:Y)_{p_{XY|Z=z}}.
\end{align*}
Using this and Theorem \ref{theorem: mibound}, we can bound $I_q(X:Y|Z)$ as 
\begin{align*}
\max\limits_{\substack{p_{XYZ}=p_Zp_{X|Z}p_{Y|Z}}} I_q(X:Y|Z) &=\max\limits_{\substack{p_{XYZ}=p_Zp_{X|Z}p_{Y|Z}}}\sum_z p_z^qI_q(X:Y)_{p_{XY|Z=z}}\\
&\leq \max\limits_{p_Z}\sum_z p_z^q \max\limits_{p_{X|Z}p_{Y|Z}}I_q(X:Y)_{p_{XY|Z=z}}\\
&=\max\limits_{p_Z}\sum_zp_z^qf(q,d_X,d_Y)=f(q,d_X,d_Y)\,.
\end{align*}
The last step holds because for all $q>1$, $\sum_zp_z^q$ is
maximized by deterministic distributions over $Z$ with a maximum value
of $1$ i.e., only distributions $p_{XYZ}$ that are deterministic over
$Z$ saturate the upper bound of $f(q,d_X,d_Y)$.  This completes the proof.
\end{proof}

Two corollaries of Theorem~\ref{theorem: causal} naturally follow.
\begin{corollary}
\label{corollary: causal} Let $X$, $Y$ and $Z$ be random variables
with fixed alphabet sizes.  Then for all $q\geq1$ we have
\begin{equation*}
   \max\limits_{\substack{p_{XYZ}\\ p_{XY|Z}=p_{X|Z}p_{Y|Z}}} I_q(X:Y|Z) = \max\limits_{\substack{p_{XY}\\ p_{XY}=p_Xp_Y}} I_q(X:Y)\,.
\end{equation*}
Furthermore, for $q>1$, the maximum on the left hand side is achieved only by distributions in which for some fixed value $k$ the joint probabilities are given by $p_{xyz}=\begin{cases}\frac{1}{d_Xd_Y} \quad \text{if} \quad z=k\\ 0
  \quad \qquad \text{otherwise}\end{cases}$, while the maximum on the right hand side occurs if and
only if $P_{XY}$ is the uniform distribution.
\end{corollary}

The significance of these new relations for causal structures is then
given by the following corollary.

\begin{corollary}
\label{corollary: mainconstraint}
Let $p_{X_1\ldots X_n}$ be a distribution compatible with the classical
causal structure $\cG^\rC$ and $X$, $Y$ and $Z$ be disjoint subsets of
$\{X_1,\ldots,X_n\}$ such that $X$ and $Y$ are d-separated given $Z$.
Then for all $q\geq1$ we have
\begin{equation*}
    I_q(X:Y|Z)\leq f(q,d_X,d_Y)\,,
\end{equation*}
where $d_X$ is the product of $d_{X_i}$ for all $X_i\in X$, and
likewise for $d_Y$.
\end{corollary}

\begin{remark}
The results of this section can be generalised to the quantum case under certain assumptions i.e., as constraints on quantum Tsallis entropies implied by certain quantum causal structures (see the Appendix~\ref{appendix: quantum} for details). Note that only constraints on the classical Tsallis entropy vectors derived in this section are required to detect the classical-quantum gap. Hence, Appendix~\ref{appendix: quantum} is not pertinent to the main results of this paper but can be seen as additional results regarding the properties of quantum Tsallis entropies.
\end{remark}

\subsection{Number of independent Tsallis entropic causal constraints}
\label{ssec: numconst}
We saw previously that in the Shannon case ($q=1$), the $n$ conditions
of the form $I(X_i:X_i^\nuparrow|X_i^{\downarrow_1})=0$
($i=1,\ldots,n$) imply all the independence relations that follow from
the causal structure. In the Tsallis case however, the $n$ conditions
of the form
$I_q(X_i:X_i^{\nuparrow}|X_i^{\downarrow_1})\leq
f(q,d_{X_i},d_{X_i^{\nuparrow}})$ do not do the same.  In the
bipartite Bell and triangle causal structures we find that there is no
redundancy amongst the 53 and 126 distinct Tsallis entropic
inequalities that are implied by the d-separation relations in the
corresponding DAGs in the case where the dimension (cardinality) of each
individual node is taken to be $d$.  In more detail, we used linear
programming to show that each implication of d-separation yields a
non-trivial entropic causal constraint for all $q>1$ and $d>2$ for the
bipartite Bell and triangle causal structures. By comparison, in these causal structures five and six independent Shannon entropic constraints imply all the others.  As an illustration of the
difference, in the Shannon case, $I(A:BC)=0$ implies $I(A:B)=I(A:C)=0$, whereas
the analogous implication does not hold in the Tsallis case in
general: although $I_q(A:BC)\leq f(q,d_A,d_{BC})$ implies
$I_q(A:B)\leq f(q,d_A,d_{BC})$, it is not the case that
$I_q(A:BC)\leq f(q,d_A,d_{BC})$ implies
$I_q(A:B)\leq f(q,d_A,d_B)$.\footnote{For an explicit counterexample,
  consider
  $p_{ABC}=\{\frac{3}{10},0.0,\frac{2}{10},0.0,\frac{1}{10},\frac{1}{10},\frac{2}{10},\frac{1}{10}\}$
  over binary $A$, $B$ and $C$ for which $I_2(A:BC)=9/25<3/8=f(2,2,4)$
  but $I_2(A:B)=13/50>1/4=f(2,2,2)$.}

The number of distinct conditional independences (and hence the number
of independent Tsallis constraints that follow from d-separation) in a
DAG depends on the specific graph, however for any DAG $\cG_n$ with
$n$ nodes, the number of such constraints can be upper bounded by that
of the $n$-node DAG where all $n$ nodes are independent i.e., the $n$
node DAG with no edges. The number of conditions in this DAG can be
thought of as the number of ways of partitioning $n$ objects into four
disjoint subsets\footnote{The four subsets correspond
to the three arguments of the conditional mutual information and a set
of `leftovers'.} such that the first two are non-empty and where the
ordering of the first two does not matter. Therefore, there are at most
$\frac{1}{2}(4^n-2\times3^n+2^n)$ such conditions.

\subsection{Using Tsallis entropies in the entropy vector method}
\label{ssec: tsalentvec}

We used the causal constraints of Corollary \ref{corollary:
  mainconstraint} in the entropy vector method with the aim of
deriving new quantum-violatable entropic inequalities for the triangle
causal structure (Figure \ref{fig: triangle}). To do so, we started
with the variables $A,B,C,X,Y,Z$ of the triangle causal structure, the
Shannon constraints and causal constraints satisfied by the Tsallis
entropy vectors over these variables (Corollary~\ref{corollary:
  mainconstraint}) and used a Fourier-Motzkin (FM) elimination
algorithm (from {\sc porta}\footnote{Polyhedral Representation and
  Transformation Algorithm: \url{http://porta.zib.de/}.}) to
eliminate the Tsallis entropy components involving the unobserved
variables $A,B,C$ and obtain the constraints on the observed
nodes $X,Y,Z$.

The Tsallis entropy vector for the six nodes has $2^6-1=63$ components. The
required marginal scenario with the observed nodes $X,Y,Z$ has Tsallis
entropy vectors with $2^3-1=7$ components and in this case, the Fourier-Motzkin algorithm has to run $56$ iterations, each of which eliminates
one variable.

Starting with the full set of 126 Tsallis entropic causal constraints
for the triangle causal structure as well as the 246 independent
Shannon constraints, the Fourier-Motzkin elimination algorithm did not
finish within several days on a standard desktop PC and the number of
intermediate inequalities generated grew to about 90,000 after 11 steps. Because of this we instead tried starting with a subset
comprising 15 of the 126 Tsallis entropic causal
constraints\footnote{These included the 6 that follow from ``each node
  $N_i$ is conditionally independent of its descendants given its
  parents" (denoted as $N_i\perp N_i^{\nuparrow}|N_i^{\downarrow_1}$)
  and 9 more chosen arbitrarily from the total of 126 independent
  Tsallis constraints we found for the triangle. The 6 former
  constraints for the triangle (Figure~\ref{fig: triangle}) are
  $A\perp CXB$, $B\perp CYA$, $C\perp BZA$, $X\perp YAZ|CB$,
  $Y\perp XBZ|AC$ and $Z\perp YCX|AB$. An example of 9 more
  constraints for which the procedure did not work are $X\perp Y|CB$,
  $X\perp A|CB$, $X\perp Z|CB$, $Y\perp X|AC$, $Y\perp B|AC$,
  $Y\perp Z|AC$, $Z\perp Y|AC$, $Z\perp C|AB$ and $Z\perp X|AB$.  We
  also tried some other choices and number of constraints but this did
  not lead to any improvement.}  i.e., 261 constraints on 63
dimensional vectors. We considered the case of $q=2$ and where the six
random variables are all binary. Again, in this case the algorithm did
not finish after several days. We also tried starting with fewer
causal constraints (for example, the six constraints analogous to the
Shannon case) as well as using a modified code, optimised to deal with
redundancies better but both of these attempts made no significant
difference to this outcome.

Such a rapid increase of the number of inequalities in each step is a
known problem with Fourier-Motzkin elimination where an elimination
step over $n$ inequalities can result in up to $n^2/4$ inequalities in
the output and running $d$ successive elimination steps can yield a
double exponential complexity of $4(n/4)^{2^d}$
\cite{Williams1986}. This rate of increase can be kept under control
when the resulting set of inequalities has many redundancies. This
happens in the Shannon case where the causal constraints are simple
equalities and the system of 246 Shannon constraints plus 6 Shannon
entropic causal constraints reduces to a system of just 91 independent
inequalities before the FM elimination. In the Tsallis case, no
reduction of the system of inequalities is possible in general due to
the nature of the causal constraints. The fact that the Tsallis
entropic causal constraints are inequality constraints rather than
equalities also contributes to the computational difficulty since each
independent equality constraint in effect reduces the dimension of the
problem by~1.

We also tried the same procedure on the bipartite Bell causal
structure (Figure \ref{fig: Bell}), again for $q=2$ and binary
variables. Here, starting with the full set of 53 causal constraints,
again resulted in the program running for over a week without nearing
the end, and a similar result was obtained when starting only with
8--10 causal constraints. While starting with fewer causal constraints
such as the 5 conditional independence constraints (one for each node)
resulted in a terminating program, no non-trivial entropic
inequalities were obtained (i.e., we only obtained constraints
corresponding to Shannon constraints or causal constraints that follow directly from $d$ separation).\footnote{For example, we were able to obtain
$I_2(A:BY)\leq \frac{7}{16}$ and $I_2(B:AX)\leq \frac{7}{16}$, while,
in the case of binary variables and $q=2$, the independences in the
DAG together with Theorem~\ref{theorem: mibound} imply
$I_2(A:BY)\leq \frac{6}{16}$ and $I_2(B:AX)\leq \frac{6}{16}$, which
are the Tsallis entropic equivalents of the two non-signalling
constraints.}

\section{New Tsallis entropic inequalities for the triangle causal structure}\label{sec: newineq}
Despite the limitations encountered in applying the entropy vector
method to Tsallis entropies (Section~\ref{ssec: tsalentvec}), here we
find new Tsallis entropic inequalities for the triangle causal
structure for all $q\geq1$ by using known inequalities for the Shannon entropy~\cite{Chaves2014} and the causal constraints derived in Section~\ref{sec: tsalcaus}. Using the entropy vector method for Shannon entropies, the following three classes of entropic inequalities were obtained for the triangle causal structure (Figure~\ref{fig: triangle}) in~\cite{Chaves2014}\footnote{Note that a tighter  entropic characterization was found in~\cite{Weilenmann2018} based on non-Shannon inequalities, and that the techniques introduced here could also be applied to these.}. Including all permutations of $X$, $Y$ and $Z$, these yield 7 inequalities.

\begin{subequations}
\begin{equation}
\label{eq: Shineq1}
    -H(X)-H(Y)-H(Z)+H(XY)+H(XZ)\geq 0,
\end{equation}
\begin{equation}
\label{eq: Shineq2}
    -5H(X)-5H(Y)-5H(Z)+4H(XY)+4H(XZ)+4H(YZ)-2H(XYZ)\geq 0,
\end{equation}
\begin{equation}
\label{eq: Shineq3}
    -3H(X)-3H(Y)-3H(Z)+2H(XY)+2H(XZ)+3H(YZ)-H(XYZ)\geq 0.
\end{equation}
\end{subequations}

By replacing the Shannon entropy $H()$ with the Tsallis entropy $S_q()$
on the left hand side of these inequalities and minimizing the
resultant expression over our outer approximation to the classical
Tsallis entropy cone for the triangle causal structure, one can obtain
valid Tsallis entropic inequalities for this causal structure. More
precisely, the outer approximation to the classical Tsallis entropy
cone for the triangle is characterised by the $6+6(6-1)2^{6-3}=246$
independent Shannon constraints (monotonicity and strong subadditivity
constraints) and the $126$ causal constraints (one for each
conditional independence implied by the causal structure). To perform
this minimization we used {\sc
  LPAssumptions}~\cite{LPAssumptions}, a linear program
solver in Mathematica that implements the simplex method allowing for
unspecified variables. In our case, we assumed that the dimensions of all the unobserved nodes
($A$,$B$ and $C$) are equal to $d_u$ and those of all the observed
nodes ($X$, $Y$ and $Z$) is $d_o$, and so the unspecified variables are $q\geq 1$,
$d_u\geq 2$ and $d_o\geq 2$. We obtained the following 
Tsallis entropic inequalities for the triangle.
\begin{subequations}
\begin{equation}
\label{eq: Tsineq1}
    -S_q(X)-S_q(Y)-S_q(Z)+S_q(XY)+S_q(XZ)\geq B_1(q,d_o,d_u), 
\end{equation}

\begin{align}
\begin{split}
\label{eq: Tsineq2}
   -5S_q(X)-5S_q(Y)-5S_q(Z)+4S_q(XY)+4S_q(XZ)+4S_q(YZ)-2S_q(XYZ)\\
    \geq B_2(q,d_o,d_u):=\max\big(B_{21}(q,d_o,d_u),B_{22}(q,d_o,d_u)\big), 
\end{split}
\end{align}

\begin{equation}
\label{eq: Tsineq3}
    -3S_q(X)-3S_q(Y)-3S_q(Z)+2S_q(XY)+2S_q(XZ)+3S_q(YZ)-S_q(XYZ)\geq B_3(q,d_o,d_u), 
\end{equation}
\end{subequations}
where,
\begin{subequations}
\begin{equation}
\label{eq: Bndineq1}
    B_1(q,d_o,d_u)=-\frac{1}{q-1}\Bigg(1-d_o^{1-q}\Bigg)\Bigg(2-d_o^{1-q}-d_u^{1-q}\Bigg),
\end{equation}

\begin{align}
\begin{split}
\label{eq: Bndineq2}
    B_{21}(q,d_o,d_u)&=-\frac{1}{q-1}\Bigg(11+d_u^{3-3q}+6d_o^{2-2q}+3d_o^{1-q}d_u^{1-q}-6d_u^{1-q}-15d_o^{1-q}\Bigg),\\
    B_{22}(q,d_o,d_u)&=-\frac{1}{q-1}\Bigg(10+d_o^{1-q}d_u^{3-3q}+5d_o^{2-2q}+2d_o^{1-q}d_u^{1-q}-5d_u^{1-q}-13d_o^{1-q}\Bigg),
\end{split}
\end{align}

\begin{equation}
\label{eq: Bndineq3}
    B_3(q,d_o,d_u)=-\frac{1}{q-1}\Bigg(6+d_o^{1-q}d_u^{2-2q}+3d_o^{2-2q}+d_o^{1-q}d_u^{1-q}-3d_u^{1-q}-8d_o^{1-q}\Bigg).
\end{equation}
\end{subequations}
Note that
$\lim_{q\rightarrow 1}B_1=\lim_{q\rightarrow 1}B_2=\lim_{q\rightarrow
  1}B_3=0$ $\forall d_u,d_o \geq 2$, recovering the original
inequalities for Shannon entropies (Equations~(\ref{eq:
  Shineq1})--(\ref{eq: Shineq3})) as a special case.

In~\cite{Rosset2017}, an upper bound on the dimensions of classical unobserved systems needed to reproduce a set of observed correlations is derived in terms of the dimensions of the observed systems. In the case of the triangle causal structure with
$d_X=d_Y=d_Z=d_o$ and $d_A=d_B=d_C=d_u$ as considered here, the result
of~\cite{Rosset2017} implies that all classical correlations $P_{XYZ}$
can be reproduced by using hidden systems of dimension at most $d_o^3-d_o$. Since the dimension
of the unobserved systems is unknown, it makes sense to take the
minimum of the derived bounds over all $d_u$ between $2$ and $d_o^3-d_o$. By
taking their derivative, one can verify that for $q>1$ each of the
functions $B_1$, $B_{21}$, $B_{22}$ and $B_3$ is monotonically
decreasing in $d_o$ and $d_u$, and hence that the minimum is obtained
for $d_u=d_o^3-d_o$ for any given $d_o\geq
2$. It follows that for all $q>1$ and $d_o\geq2$ relations of the same
form as Equations~\eqref{eq: Tsineq1}--\eqref{eq: Tsineq3} hold, with
the quantities on the right hand sides replaced by
\begin{subequations}
\begin{align}
\begin{split}
\label{eq: B2ndineq1}
   B_1^{*}(q,d_o)&=B_1(q,d_o,d_o^3-d_o)\\
  & =-\frac{1}{q-1}\Bigg(2+d_o^{2-2q}-3d_o^{1-q}+d_o(-d_o + d_o^3)^{-q}-d_o^3 (-d_o + d_o^3)^{-q}-d_o^{2 - q} (-d_o + d_o^3)^{-q}+d_o^{4 - q} (-d_o + d_o^3)^{-q}\Bigg)
  \end{split}
\end{align}

\begin{align}
\begin{split}
\label{eq: B2ndineq2}
    B_{21}^{*}(q,d_o)&=B_{21}(q,d_o,d_o^3-d_o)\\
    &=-\frac{1}{q-1}\Bigg(11 + 6 d_o^{2 - 2 q} - 
 15 d_o^{1 - q} + (-d_o + d_o^3)^{3 - 3 q} - 
 6 (-d_o + d_o^3)^{1 - q} + 
 3 d_o^{1 - q} (-d_o + d_o^3)^{1 - q}\Bigg),\\
    B_{22}^{*}(q,d_o)&=B_{22}(q,d_o,d_o^3-d_o)\\
    &=-\frac{1}{q-1}\Bigg(10 + 5 d_o^{2 - 2 q} - 13 d_o^{1 - q} + 
 d_o^{1 - q} (-d_o + d_o^3)^{3 - 3 q} - 
 5 (-d_o + d_o^3)^{1 - q} + 
 2 d_o^{1 - q} (-d_o + d_o^3)^{1 - q}\Bigg),
\end{split}
\end{align}

\begin{equation}
\label{eq: B2ndineq3}
    B_3^{*}(q,d_o)=B_3(q,d_o,d_o^3-d_o)=-\frac{1}{q-1}\Bigg(6 + 3 d_o^{2 - 2 q} - 8 d_o^{1 - q} + 
 d_o^{1 - q} (-d_o + d_o^3)^{2 - 2 q}- 
 3 (-d_o + d_o^3)^{1 - q} + 
 d_o^{1 - q} (-d_o + d_o^3)^{1 - q}\Bigg)\,.
\end{equation}
\end{subequations}
A quantum violation of any of these bounds would imply that no
unobserved classical systems of arbitrary dimension could
reproduce those quantum correlations.

\begin{remark}
\label{remark: boundlimit}
Because they are monotonically decreasing, the bounds for $d_u=d_o^3-d_o$ are not as tight as the $d_u$-dependent bounds for general $q>1$. Nevertheless, as $q\to 1$, all the bounds $B^*(q,d_o)$ tend to 0, reproducing the known result of~\cite{Fritz13} for the Shannon case.
\end{remark}

\begin{remark}
In some cases it may be interesting to show quantum violations of these inequalities for low values of $d_u$, hence ruling out classical explanations with hidden systems of low dimensions, while possibly leaving open the case of arbitrary classical explanations. This would be interesting if it could be established that using hidden quantum systems allows for much lower dimensions than for hidden classical systems, for example.
\end{remark}

\subsection{Looking for quantum violations}\label{ssec: qviol}
It is known that the triangle causal structure (Figure~\ref{fig:
  triangle}) admits non-classical correlations such as Fritz's
distribution~\cite{Fritz2012}. The idea behind this distribution is to
embed the CHSH game in the triangle causal structure such that
non-locality for the triangle follows from the non-locality of the
CHSH game. To do so, $C$ is replaced by the sharing of a maximally
entangled pair of qubits, and $A$ and $B$ are taken to be uniformly
random classical bits. The observed variables $X$, $Y$ and $Z$ in
Figure~\ref{fig: triangle} are taken to be pairs of the form
$X:=(\tilde{X},B)$, $Y:=(\tilde{Y},A)$ and $Z:=(A,B)$, where $\tilde{X}$
and $\tilde{Y}$ are generated by measurements on the halves of the
entangled pair with $B$ and $A$ used to choose the settings such that
the joint distribution $P_{\tilde{X}\tilde{Y}|BA}$ maximally violates
a CHSH inequality.  By a similar post-processing of other non-local
distributions in the bipartite Bell causal structure (Figure~\ref{fig:
  Bell}) such as the Mermin-Peres magic square game~\cite{Mermin1990,
  Peres1990} and chained Bell inequalities~\cite{BraunsteinCaves88},
one can obtain other non-local distributions in the triangle that
cannot be reproduced using classical systems.  We explore whether any
of these violate any of our new inequalities.

Since the values of $B_i(q,d_o,d_u)$ are monotonically decreasing in
$d_o$ and $d_u$, if a distribution realisable in a quantum causal
structure does not violate the bounds~\eqref{eq: Tsineq1}--\eqref{eq:
  Tsineq3} for all $q\geq1$ and some fixed values of $d_o$ and $d_u$,
then no violations are possible for $d_o'>d_o$, $d_u'>d_u$. We
therefore take the smallest possible values of $d_o$ and $d_u$ when
showing that a particular distribution cannot violate any of the bounds.
  
For Fritz's distribution~\cite{Fritz2012}, $C$ is a two-qubit maximally
entangled state, $A$ and $B$ are binary random variables while $X$, $Y$ and $Z$ are random variables of dimension 4, i.e., the actual observed dimensions are
$(d_X,d_Y,d_Z)=(4,4,4)$ in this case. Here we see that taking
$d_o=4$ and the smallest possible $d_u$ which is $d_u=2$, the left
hand sides of Equations~\eqref{eq: Tsineq1}--\eqref{eq: Tsineq3}
evaluated for Fritz's distribution do not violate the corresponding bounds $B_i(q,d_o=4,d_u=2)$ for any
$q\geq 1$. This means that it is not possible to detect any quantum
advantage of this distribution (even over the case where the
unobserved systems are classical bits) using this method, and
automatically implies that it cannot violate the bounds
$B_i(q,d_o=4,d_u)$ for $d_u\geq 2$.

We also considered the chained Bell and magic square correlations
embedded in the triangle causal structure analogously to the case
discussed above.  For each of these, we define $d^i$ to be the
smallest value of $d_o$ for which the bound $B_i(q,d_o=d^i,d_u=2)$
cannot be violated for any $q>1$. The values of $d^i$ are given in
Table~\ref{table: violations} for the different cases of the chained
Bell correlations and the magic square. Since the values of $d^i$ are
always lower than the smallest of the observed dimensions in the
problem, and due to the monotonicity of the bounds it follows that
none of these quantum distributions violate any of our inequalities
when the observed dimension is set to $d_o^{\min}$.

\begin{table}[t]
\begin{center}
\begin{tabular}{ ccccc } 
\hline
\hline
\quad & \multicolumn{3}{c}{$d^i$} & smallest\\
\cline{2-4}
Scenario & Ineq.~\eqref{eq: Tsineq1} ($i=1$)\ \ \ \  & Ineq.~\eqref{eq: Tsineq2} ($i=2$)\ \ \ \  & Ineq.~\eqref{eq: Tsineq3} ($i=3$)\ & \ \ observed dim. ($d_o^{\min}$)\\
 \hline
$N=2$ & 2 & 2 & 2 & 4 \\ 
$N=3$ & 3 & 2 & 3 & 6 \\ 
$N=4$ & 4 & 2 & 4 & 8 \\ 
$N=5$ & 5 & 2 & 5 & 10 \\ 
$N=6$ & 6 & 2 & 6 & 12 \\ 
$N=7$ & 7 & 2 & 7 & 14 \\ 
$N=8$ & 8 & 2 & 8 & 16 \\ 
$N=9$ & 9 & 3 & 9 & 18 \\ 
$N=10$ & 10 & 3 & 10 & 20 \\ 
Magic Sq. & 4 & 2 & 4 & 9 \\ 
 \hline
 \hline
\end{tabular}
\end{center}
\caption{\textbf{Values of $d^i$ for the chained Bell and magic square
    correlations embedded in the triangle causal structure.} The
  values of $N$ correspond to the number of inputs per party in the
  chained Bell inequality, which always has two outputs per party (the
  $N=2$ case corresponds to Fritz's
  distribution~\cite{Fritz2012}). When embedded in the triangle, the
  number of outcomes of the observed nodes are
  $(d_X,d_Y,d_Z)=(2N,2N,N^2)$.  The last column of the table gives the
  minimum of the observed node dimensions $(d_X,d_Y,d_Z)$ for each
  $N$, which is simply $2N$. For the magic square, the dimensions
  $(d_X,d_Y,d_Z)$ are $(12,12,9)$. In all cases, the minimum value of
  $d^i$ such that the Inequalities~(\ref{eq: Tsineq1})--(\ref{eq:
    Tsineq3}) with bounds $B_i(q,d_o=d^i,d_u=2)$ are not violated for
  any $q\geq 1$ is less than the minimum observed dimension
  $d_o^{\min}$, and hence no violations of~\eqref{eq:
    Tsineq1}--\eqref{eq: Tsineq3} could be found for the relevant case
  with $d_o=d^{\min}_o$.}
\label{table: violations}
\end{table}

We further checked for violations of Inequalities~(\ref{eq:
  Tsineq1})--(\ref{eq: Tsineq3}) by sampling random quantum states for
the systems $A$, $B$ and $C$ and random quantum measurements whose
outcomes would correspond to the classical variables $X$, $Y$ and
$Z$. The value of $q$ was also sampled randomly between $1$ and
$100$. We considered the cases where the shared systems were pairs of
qubits with 4 outcome measurements ($d_X=d_Y=d_Z=4$) and qutrits with
9 outcome measurements ($d_X=d_Y=d_Z=9$) but were unable to find
violations of any of the inequalities even for the bounds with the
$d_o=4, d_u=2$ (two qubit case) and $d_o=9, d_u=2$ (two qutrit case),
i.e., the bounds obtained when the unobserved systems are classical
bits.

\begin{remark}
In the derivation of Inequalities~(\ref{eq: Tsineq1})--(\ref{eq:
  Tsineq3}), we set the dimensions of the observed nodes
$X$, $Y$ and $Z$ to all be equal and those of the unobserved
nodes $A$, $B$ and $C$ to also all be equal. One could in principle
repeat the same procedure taking different dimensions for all 6
variables but we found the computational procedure too
demanding. However, Table~\ref{table: violations} shows that even when we consider the bounds $B_i(q,d_o,d_u)$ with $d_o$ and $d_u$ much smaller than the actual dimensions, known non-local distributions in the triangle considered in Table~\ref{table: violations} do not violate the corresponding Inequalities~(\ref{eq: Tsineq1})--(\ref{eq:
  Tsineq3}) for any $q\geq 1$. Since the bounds are monotonically decreasing in $d_u$ and $d_o$, even if we obtained the general bounds for arbitrary dimensions of $X$, $Y$, $Z$, $A$, $B$ and $C$, they would be strictly weaker than $B_i(q,d^i,d_u=2)$ $\forall i\in \{1,2,3\}, q\geq 1$ and can certainly not be violated by these distributions.
\end{remark}

\section{Discussion}\label{sec: causaldiscussion}
We have investigated the use of Tsallis entropies within the entropy
vector method to causal structures, showing how causal constraints
imply bounds on the Tsallis entropies of the variables
involved. Although Tsallis entropies for $q\geq1$ possess many
properties that aid their use in the entropy vector method, the nature
of the causal constraints makes the problem significantly more
computationally challenging than in the case of Shannon entropy.  This
meant that we were unable to complete the desired computations in the
former case, even for some of the simplest causal structures.
Nevertheless, we were able to derive new classical causal constraints
expressed in terms of Tsallis entropy by analogy with known Shannon
constraints, but were unable to find cases where these
were violated, even using quantum distributions that are known not to
be classically realisable.  This mirrors an analogous result for
Shannon entropies~\cite{Weilenmann2018}.

Tsallis entropies are known to give improvements~\cite{Wajs15} in
cases that involve post-selection.  While post-selection cannot be
used for general causal structures (including the triangle), it would
be interesting to understand whether using Tsallis entropy helps in
other cases for which post-selection is applicable.

One could also investigate whether other entropic quantities could be
used in a similar way. The R\'enyi entropies of order $\alpha$ do not
satisfy strong subadditivity for $\alpha \neq 0,1$, while the R\'enyi
as well as the min and max entropies fail to obey the chain rules for
conditional entropies.  Thus, use of these in the entropy vector method, would
require an entropy vector with components for all possible conditional
entropies as well as unconditional ones, considerably increasing the
dimensionality of the problem, which we would expect to make the
computations harder.\footnote{In some cases, not having a chain rule
  may not be prohibitive~\cite{WeilenmannGPT}.}
  
Further, one could consider using algorithms other than Fourier-Motzkin
elimination to obtain non-trivial Tsallis entropic constraints
over observed nodes starting from the Tsallis cone over all the nodes
(see e.g.,~\cite{Gl_le_2018}). These could in principle yield
solutions even in cases where FM elimination becomes
intractable. However, we found that the FM elimination procedure
became intractable even when starting out with only a small subset of
the Tsallis entropic causal constraints for a simple causal structure
such as the Bell one. 
This suggests that the difficulty is not only with the number of constraints, but also with their nature (in particular, that they are not equalities and depend non-trivially on the dimensions). Consequently, we bypassed FM elimination and used an alternative technique to obtain new Tsallis entropic inequalities for the Triangle causal structure (Section~\ref{sec: newineq}). 

It is also worth noting that the following alternative definition of the Tsallis
conditional entropy was proposed in~\cite{ABE2001157}.
\begin{equation}
    \tilde{S_q}(X|Y)=\frac{1}{1-q}\frac{\sum_yp_y^qS_q(X|Y=y)}{\sum_y p_y^q}=\frac{1}{1-q}\left(\frac{\sum_{x,y}p_{xy}^q}{\sum_y p_y^q}-1\right).
\end{equation}
Using this definition, Tsallis entropies would satisfy the same causal
constraints as the Shannon entropy (Equation~\eqref{eq:
  shancausmain}). However, the conditional entropies defined this way
do not satisfy the chain rules of Equation~\eqref{eq: chain2} but
instead obey a non-linear chain rule,
$S_q(XY)=S_q(X)+S_q(Y|X)+(1-q)S_q(X)S_q(Y|X)$~\cite{ABE2001157}. This
would again mean that conditional entropies would need to be included
in the entropy vector.  Furthermore, since Fourier-Motzkin elimination
only works for linear constraints, an alternative algorithm would be
required to use this chain rule in conjunction with the entropy vector
method.

That the inequalities for Tsallis entropy derived in this work depend
on the dimensions of the systems involved could be used to certify
that particular observed correlations in a classical causal structure
require a certain minimal dimension of unobserved systems to be
realisable. To show this would require showing that
classically-realisable correlations violate one of the
inequalities for some $d_u$.  Such bounds would then complement
the upper bounds of~\cite{Rosset2017}.  However, in some cases we know our bounds are not tight enough to do this.  As a simple example, within the triangle causal structure we tried taking $X=(X_B,X_C)$, $Y=(Y_A,Y_C)$ and $Z=(Z_A,Z_B)$ with $X_B=Z_B$, $X_C=Y_C$ and $Y_A=Z_A$ where each are uniformly distributed with cardinality $D$, for $D\in\{3,\ldots,10\}$.  In this case it is clear that the correlations cannot be achieved with classical unobserved systems with $d_u=2$. Taking the bound with $d_u=2$ and $d_o=D^2$ no violations of~\eqref{eq: Tsineq1}--\eqref{eq: Tsineq3} were seen by plotting the graphs for $q\in[1,20]$, for the range of $D$ above. Hence, our bounds are too loose to certify lower bounds on $d_u$ in this case.

While our analysis highlights significant drawbacks of using Tsallis
entropies for analysing causal structures, it does not rule out the
possibility of Tsallis entropies being able to detect the
classical-quantum gap\footnote{Proving that Tsallis entropies are
  unable to do this would also be difficult.  For instance, the proof
  of~\cite{Weilenmann16} that Shannon entropies are unable to detect
  the gap in line-like causal structures involves first characterising
  the marginal polytope through Fourier-Motzkin elimination, which
  itself proved to be computationally infeasible with Tsallis
  entropies even for the simplest line-like causal structure, the
  bipartite Bell scenario.} in these causal structures, or others.  To overcome the difficulties we encountered we would either need increased computational power, or the use of new, alternative techniques for analysing causal structures (with or without entropies).

\begin{acknowledgements}
  We thank Mirjam Weilenmann and Elie Wolfe for useful discussions and
  also thank Mirjam Weilenmann for sharing some Mathematica code. VV
  acknowledges financial support from the Department of Mathematics,
  University of York. RC is supported by EPSRC's Quantum
  Communications Hub (grant number EP/M013472/1) and by an EPSRC First
  Grant (grant number EP/P016588/1).
\end{acknowledgements}



\begin{thebibliography}{10}
\expandafter\ifx\csname url\endcsname\relax
  \def\url#1{\texttt{#1}}\fi
\expandafter\ifx\csname urlprefix\endcsname\relax\def\urlprefix{URL }\fi
\providecommand{\bibinfo}[2]{#2}
\providecommand{\eprint}[2][]{\url{#2}}

\bibitem{Bell}
\bibinfo{author}{Bell, J.~S.}
\newblock \emph{\bibinfo{title}{Speakable and unspeakable in quantum
  mechanics}} (\bibinfo{publisher}{Cambridge University Press},
  \bibinfo{year}{1987}).
\newblock \bibinfo{note}{\url{https://doi.org/10.1017/CBO9780511815676}}.

\bibitem{Ekert91}
\bibinfo{author}{Ekert, A.~K.}
\newblock \bibinfo{title}{Quantum cryptography based on {B}ell's theorem}.
\newblock \emph{\bibinfo{journal}{Phys. Rev. Lett.}}
  \textbf{\bibinfo{volume}{67}}, \bibinfo{pages}{661--663}
  (\bibinfo{year}{1991}).
\newblock
  \bibinfo{note}{\url{https://link.aps.org/doi/10.1103/PhysRevLett.67.661}}.

\bibitem{MayersYao}
\bibinfo{author}{Mayers, D.} \& \bibinfo{author}{Yao, A.}
\newblock \bibinfo{title}{Quantum cryptography with imperfect apparatus}.
\newblock In \emph{\bibinfo{booktitle}{Proceedings of the 39th Annual Symposium
  on Foundations of Computer Science (FOCS-98)}}, \bibinfo{pages}{503--509}
  (\bibinfo{publisher}{IEEE Computer Society}, \bibinfo{address}{Los Alamitos,
  CA, USA}, \bibinfo{year}{1998}).
\newblock
  \bibinfo{note}{\url{http://doi.ieeecomputersociety.org/10.1109/SFCS.1998.743501}}.

\bibitem{BHK}
\bibinfo{author}{Barrett, J.}, \bibinfo{author}{Hardy, L.} \&
  \bibinfo{author}{Kent, A.}
\newblock \bibinfo{title}{No signalling and quantum key distribution}.
\newblock \emph{\bibinfo{journal}{Physical Review Letters}}
  \textbf{\bibinfo{volume}{95}}, \bibinfo{pages}{010503}
  (\bibinfo{year}{2005}).
\newblock \bibinfo{note}{\url{https://doi.org/10.1103/PhysRevLett.95.010503}}.

\bibitem{RogerThesis}
\bibinfo{author}{Colbeck, R.}
\newblock \bibinfo{title}{Quantum and relativistic protocols for secure
  multi-party computation.}
\newblock \emph{\bibinfo{journal}{PhD Dissertation, University of Cambridge}}
  (\bibinfo{year}{2007}).
\newblock \bibinfo{note}{\url{https://arxiv.org/abs/0911.3814}}.

\bibitem{CK}
\bibinfo{author}{Colbeck, R.} \& \bibinfo{author}{Kent, A.}
\newblock \bibinfo{title}{Private randomness expansion with untrusted devices}.
\newblock \emph{\bibinfo{journal}{Journal of Physics A}}
  \textbf{\bibinfo{volume}{44}}, \bibinfo{pages}{095305}
  (\bibinfo{year}{2011}).
\newblock
  \bibinfo{note}{\url{https://iopscience.iop.org/article/10.1088/1751-8113/44/9/095305}}.

\bibitem{Pironio2009}
\bibinfo{author}{Pironio, S.} \emph{et~al.}
\newblock \bibinfo{title}{Random numbers certified by {B}ell's theorem}.
\newblock \emph{\bibinfo{journal}{Nature}} \textbf{\bibinfo{volume}{464}},
  \bibinfo{pages}{1021--4} (\bibinfo{year}{2010}).
\newblock \bibinfo{note}{\url{https://www.nature.com/articles/nature09008}}.

\bibitem{Pitowski89}
\bibinfo{author}{Pitowski, I.}
\newblock \bibinfo{title}{Quantum probability -- quantum logic}.
\newblock \emph{\bibinfo{journal}{Springer-Verlag Berlin Heidelberg}}
  \textbf{\bibinfo{volume}{321}} (\bibinfo{year}{1989}).
\newblock \bibinfo{note}{\url{https://www.springer.com/gp/book/9783662137352}}.

\bibitem{Masanes2002}
\bibinfo{author}{Masanes, L.}
\newblock \bibinfo{title}{Tight {B}ell inequality for d-outcome measurements
  correlations}.
\newblock \emph{\bibinfo{journal}{Quantum information and computation}}
  \textbf{\bibinfo{volume}{3}}, \bibinfo{pages}{345} (\bibinfo{year}{2002}).
\newblock \bibinfo{note}{\url{https://dl.acm.org/citation.cfm?id=2011532}}.

\bibitem{Bancal2010}
\bibinfo{author}{Bancal, J.-D.}, \bibinfo{author}{Gisin, N.} \&
  \bibinfo{author}{Pironio, S.}
\newblock \bibinfo{title}{Looking for symmetric {B}ell inequalities}.
\newblock \emph{\bibinfo{journal}{Journal of Physics A: Mathematical and
  Theoretical}} \textbf{\bibinfo{volume}{43}}, \bibinfo{pages}{385303}
  (\bibinfo{year}{2010}).
\newblock
  \bibinfo{note}{\url{https://iopscience.iop.org/article/10.1088/1751-8113/43/38/385303}}.

\bibitem{Cope2019}
\bibinfo{author}{Cope, T.} \& \bibinfo{author}{Colbeck, R.}
\newblock \bibinfo{title}{{B}ell inequalities from no-signaling distributions}.
\newblock \emph{\bibinfo{journal}{Physical Review A}}
  \textbf{\bibinfo{volume}{100}}, \bibinfo{pages}{22114}
  (\bibinfo{year}{2019}).
\newblock \urlprefix\url{https://doi.org/10.1103/PhysRevA.100.022114}.

\bibitem{wolfe2016}
\bibinfo{author}{Wolfe, E.}, \bibinfo{author}{Spekkens, R.~W.} \&
  \bibinfo{author}{Fritz, T.}
\newblock \bibinfo{title}{The inflation technique for causal inference with
  latent variables}  (\bibinfo{year}{2016}).
\newblock \bibinfo{note}{\url{https://arxiv.org/abs/1609.00672}}.

\bibitem{Navascues2017}
\bibinfo{author}{Navascues, M.} \& \bibinfo{author}{Wolfe, E.}
\newblock \bibinfo{title}{The inflation technique solves completely the
  classical inference problem} (\bibinfo{year}{2017}).
\newblock \bibinfo{note}{\url{https://arxiv.org/abs/1707.06476}}.

\bibitem{Yeung97}
\bibinfo{author}{Yeung, R.~W.}
\newblock \bibinfo{title}{A framework for linear information inequalities}.
\newblock \emph{\bibinfo{journal}{IEEE Transactions on Information Theory}}
  \textbf{\bibinfo{volume}{43}}, \bibinfo{pages}{1924--1934}
  (\bibinfo{year}{1997}).
\newblock \bibinfo{note}{\url{https://dl.acm.org/citation.cfm?id=2265248}}.

\bibitem{Fritz13}
\bibinfo{author}{Fritz, T.} \& \bibinfo{author}{Chaves, R.}
\newblock \bibinfo{title}{Entropic inequalities and marginal problems}.
\newblock \emph{\bibinfo{journal}{IEEE Transactions on Information Theory}}
  \textbf{\bibinfo{volume}{59}}, \bibinfo{pages}{803--817}
  (\bibinfo{year}{2013}).
\newblock \bibinfo{note}{\url{https://ieeexplore.ieee.org/document/6336823}}.

\bibitem{Chaves13}
\bibinfo{author}{Chaves, R.}
\newblock \bibinfo{title}{Entropic inequalities as a necessary and sufficient
  condition to noncontextuality and locality}.
\newblock \emph{\bibinfo{journal}{Phys. Rev. A}} \textbf{\bibinfo{volume}{87}},
  \bibinfo{pages}{022102} (\bibinfo{year}{2013}).
\newblock
  \bibinfo{note}{\url{https://link.aps.org/doi/10.1103/PhysRevA.87.022102}}.

\bibitem{Weilenmann20170483}
\bibinfo{author}{Weilenmann, M.} \& \bibinfo{author}{Colbeck, R.}
\newblock \bibinfo{title}{Analysing causal structures with entropy}.
\newblock \emph{\bibinfo{journal}{Proceedings of the Royal Society of London A:
  Mathematical, Physical and Engineering Sciences}}
  \textbf{\bibinfo{volume}{473}} (\bibinfo{year}{2017}).
\newblock
  \bibinfo{note}{\url{http://rspa.royalsocietypublishing.org/content/473/2207/20170483}}.

\bibitem{Weilenmann16}
\bibinfo{author}{Weilenmann, M.} \& \bibinfo{author}{Colbeck, R.}
\newblock \bibinfo{title}{Inability of the entropy vector method to certify
  nonclassicality in linelike causal structures}.
\newblock \emph{\bibinfo{journal}{Phys. Rev. A}} \textbf{\bibinfo{volume}{94}},
  \bibinfo{pages}{042112} (\bibinfo{year}{2016}).
\newblock
  \bibinfo{note}{\url{https://link.aps.org/doi/10.1103/PhysRevA.94.042112}}.

\bibitem{BraunsteinCaves88}
\bibinfo{author}{Braunstein, S.~L.} \& \bibinfo{author}{Caves, C.~M.}
\newblock \bibinfo{title}{Information-theoretic {B}ell inequalities}.
\newblock \emph{\bibinfo{journal}{Phys. Rev. Lett.}}
  \textbf{\bibinfo{volume}{61}}, \bibinfo{pages}{662--665}
  (\bibinfo{year}{1988}).
\newblock
  \bibinfo{note}{\url{https://link.aps.org/doi/10.1103/PhysRevLett.61.662}}.

\bibitem{Weilenmann2018}
\bibinfo{author}{Weilenmann, M.} \& \bibinfo{author}{Colbeck, R.}
\newblock \bibinfo{title}{Non-{S}hannon inequalities in the entropy vector
  approach to causal structures}.
\newblock \emph{\bibinfo{journal}{{Quantum}}} \textbf{\bibinfo{volume}{2}},
  \bibinfo{pages}{57} (\bibinfo{year}{2018}).
\newblock \bibinfo{note}{\url{https://doi.org/10.22331/q-2018-03-14-57}}.

\bibitem{Chaves2014}
\bibinfo{author}{Chaves, R.}, \bibinfo{author}{Luft, L.} \&
  \bibinfo{author}{Gross, D.}
\newblock \bibinfo{title}{Causal structures from entropic information: geometry
  and novel scenarios}.
\newblock \emph{\bibinfo{journal}{New Journal of Physics}}
  \textbf{\bibinfo{volume}{16}}, \bibinfo{pages}{043001}
  (\bibinfo{year}{2014}).
\newblock
  \bibinfo{note}{\url{https://iopscience.iop.org/article/10.1088/1367-2630/16/4/043001}}.

\bibitem{renyi}
\bibinfo{author}{R\'enyi, A.}
\newblock \bibinfo{title}{On measures of information and entropy}.
\newblock In \emph{\bibinfo{booktitle}{Proceedings of the 4th Berkeley
  Symposium on Mathematics, Statistics and Probability}},
  vol.~\bibinfo{volume}{1}, \bibinfo{pages}{547--561} (\bibinfo{year}{1961}).
\newblock
  \bibinfo{note}{\url{http://digitalassets.lib.berkeley.edu/math/ucb/text/math_s4_v1_article-27.pdf}}.

\bibitem{Wajs15}
\bibinfo{author}{Wajs, M.}, \bibinfo{author}{Kurzy\ifmmode~\acute{n}\else
  \'{n}\fi{}ski, P.} \& \bibinfo{author}{Kaszlikowski, D.}
\newblock \bibinfo{title}{Information-theoretic {B}ell inequalities based on
  {T}sallis entropy}.
\newblock \emph{\bibinfo{journal}{Phys. Rev. A}} \textbf{\bibinfo{volume}{91}},
  \bibinfo{pages}{012114} (\bibinfo{year}{2015}).
\newblock
  \bibinfo{note}{\url{https://link.aps.org/doi/10.1103/PhysRevA.91.012114}}.

\bibitem{Zhang1997ANC}
\bibinfo{author}{Zhang, Z.} \& \bibinfo{author}{Yeung, R.~W.}
\newblock \bibinfo{title}{A non-{S}hannon-type conditional inequality of
  information quantities}.
\newblock \emph{\bibinfo{journal}{IEEE Trans. Information Theory}}
  \textbf{\bibinfo{volume}{43}}, \bibinfo{pages}{1982--1986}
  (\bibinfo{year}{1997}).
\newblock \bibinfo{note}{\url{https://ieeexplore.ieee.org/document/641561}}.

\bibitem{Pearl2000}
\bibinfo{author}{Pearl, J.}
\newblock \emph{\bibinfo{title}{Causality: Models, Reasoning, and Inference}}
  (\bibinfo{publisher}{Cambridge University Press}, \bibinfo{address}{New York,
  NY, USA}, \bibinfo{year}{2000}).
\newblock \bibinfo{note}{\url{https://dl.acm.org/citation.cfm?id=1642718}}.

\bibitem{Geiger1987}
\bibinfo{author}{Geiger, D.}
\newblock \bibinfo{title}{Towards the formalization of informational
  dependencies}.
\newblock \emph{\bibinfo{journal}{Tech. rep. 880053. UCLA Computer Science}}
  (\bibinfo{year}{1987}).
\newblock \bibinfo{note}{\url{http://fmdb.cs.ucla.edu/Treports/880053.pdf}}.

\bibitem{Verma2013}
\bibinfo{author}{Verma, T.} \& \bibinfo{author}{Pearl, J.}
\newblock \bibinfo{title}{Causal networks: Semantics and expressiveness}.
\newblock \emph{\bibinfo{journal}{Proceedings of the Fourth Workshop on
  Uncertainty in Artificial Intelligence}} \textbf{\bibinfo{volume}{4}}
  (\bibinfo{year}{2013}).
\newblock \bibinfo{note}{\url{https://dl.acm.org/citation.cfm?id=719420}}.

\bibitem{Williams1986}
\bibinfo{author}{Williams, H.~P.}
\newblock \bibinfo{title}{Fourier's method of linear programming and its dual}.
\newblock \emph{\bibinfo{journal}{The American Mathematical Monthly}}
  \textbf{\bibinfo{volume}{93}}, \bibinfo{pages}{681--695}
  (\bibinfo{year}{1986}).
\newblock \bibinfo{note}{\url{https://doi.org/10.1080/00029890.1986.11971923}}.

\bibitem{Fritz2012}
\bibinfo{author}{Fritz, T.}
\newblock \bibinfo{title}{Beyond {B}ell's theorem: correlation scenarios}.
\newblock \emph{\bibinfo{journal}{New Journal of Physics}}
  \textbf{\bibinfo{volume}{14}}, \bibinfo{pages}{103001}
  (\bibinfo{year}{2012}).
\newblock
  \bibinfo{note}{\url{https://iopscience.iop.org/article/10.1088/1367-2630/14/10/103001}}.

\bibitem{Tsallis1988}
\bibinfo{author}{Tsallis, C.}
\newblock \bibinfo{title}{Possible generalization of {B}oltzmann-{G}ibbs
  statistics}.
\newblock \emph{\bibinfo{journal}{Journal of Statistical Physics}}
  \textbf{\bibinfo{volume}{52}}, \bibinfo{pages}{479--487}
  (\bibinfo{year}{1988}).
\newblock \bibinfo{note}{\url{https://doi.org/10.1007/BF01016429}}.

\bibitem{Furuichi04}
\bibinfo{author}{Furuichi, S.}
\newblock \bibinfo{title}{Information theoretical properties of {T}sallis
  entropies}.
\newblock \emph{\bibinfo{journal}{Journal of Mathematical Physics}}
  \textbf{\bibinfo{volume}{47}} (\bibinfo{year}{2004}).
\newblock
  \bibinfo{note}{\url{https://aip.scitation.org/doi/full/10.1063/1.2165744}}.

\bibitem{ABE2001157}
\bibinfo{author}{Abe, S.} \& \bibinfo{author}{Rajagopal, A.}
\newblock \bibinfo{title}{Nonadditive conditional entropy and its significance
  for local realism}.
\newblock \emph{\bibinfo{journal}{Physica A: Statistical Mechanics and its
  Applications}} \textbf{\bibinfo{volume}{289}}, \bibinfo{pages}{157 -- 164}
  (\bibinfo{year}{2001}).
\newblock
  \bibinfo{note}{\url{http://dx.doi.org/10.1016/S0378-4371(00)00476-3}}.

\bibitem{Curado1991}
\bibinfo{author}{Curado, E. M.~F.} \& \bibinfo{author}{Tsallis, C.}
\newblock \bibinfo{title}{Generalized statistical mechanics: connection with
  thermodynamics}.
\newblock \emph{\bibinfo{journal}{Journal of Physics A}}
  \textbf{\bibinfo{volume}{24}}, \bibinfo{pages}{1019--1019}
  (\bibinfo{year}{1991}).
\newblock
  \bibinfo{note}{\url{https://iopscience.iop.org/article/10.1088/0305-4470/24/2/004}}.

\bibitem{Furuichirelentropy}
\bibinfo{author}{Furuichi, S.}, \bibinfo{author}{Yanagi, K.} \&
  \bibinfo{author}{Kuriyama, K.}
\newblock \bibinfo{title}{Fundamental properties of {T}sallis relative
  entropy}.
\newblock \emph{\bibinfo{journal}{Journal of Mathematical Physics}}
  \textbf{\bibinfo{volume}{45}} (\bibinfo{year}{2004}).
\newblock
  \bibinfo{note}{\url{https://aip.scitation.org/doi/10.1063/1.1805729}}.

\bibitem{Daroczy1970}
\bibinfo{author}{Dar{\'{o}}czy, Z.}
\newblock \bibinfo{title}{{Generalized information functions}}.
\newblock \emph{\bibinfo{journal}{Information and Control}}
  \textbf{\bibinfo{volume}{16}}, \bibinfo{pages}{36--51}
  (\bibinfo{year}{1970}).
\newblock
  \urlprefix\url{http://www.sciencedirect.com/science/article/pii/S0019995870800407}.

\bibitem{LPAssumptions}
\bibinfo{author}{Colbeck, R.} \& \bibinfo{author}{Vilasini, V.}
\newblock \bibinfo{title}{L{PA}ssumptions ({M}athematica package)}.
\newblock \urlprefix\url{https://github.com/rogercolbeck/LPAssumptions}.

\bibitem{Rosset2017}
\bibinfo{author}{Rosset, D.}, \bibinfo{author}{Gisin, N.} \&
  \bibinfo{author}{Wolfe, E.}
\newblock \bibinfo{title}{Universal bound on the cardinality of local hidden
  variables in networks}.
\newblock \emph{\bibinfo{journal}{Quantum Information and Computation}}
  \textbf{\bibinfo{volume}{18}} (\bibinfo{year}{2017}).

\bibitem{Mermin1990}
\bibinfo{author}{Mermin, N.~D.}
\newblock \bibinfo{title}{Simple unified form for the major no-hidden-variables
  theorems}.
\newblock \emph{\bibinfo{journal}{Phys. Rev. Lett.}}
  \textbf{\bibinfo{volume}{65}}, \bibinfo{pages}{3373--3376}
  (\bibinfo{year}{1990}).
\newblock
  \bibinfo{note}{\url{https://link.aps.org/doi/10.1103/PhysRevLett.65.3373}}.

\bibitem{Peres1990}
\bibinfo{author}{Peres, A.}
\newblock \bibinfo{title}{Incompatible results of quantum measurements}.
\newblock \emph{\bibinfo{journal}{Physics Letters A}}
  \textbf{\bibinfo{volume}{151}}, \bibinfo{pages}{107--108}
  (\bibinfo{year}{1990}).
\newblock
  \bibinfo{note}{\url{http://www.sciencedirect.com/science/article/pii/037596019090172K}}.

\bibitem{WeilenmannGPT}
\bibinfo{author}{Weilenmann, M.} \& \bibinfo{author}{Colbeck, R.}
\newblock \bibinfo{title}{Analysing causal structures in generalised
  probabilistic theories} (\bibinfo{year}{2018}).
\newblock \bibinfo{note}{\url{https://arxiv.org/abs/1812.04327}}.

\bibitem{Gl_le_2018}
\bibinfo{author}{Gl{\"a}{\ss}le, T.}, \bibinfo{author}{Gross, D.} \&
  \bibinfo{author}{Chaves, R.}
\newblock \bibinfo{title}{Computational tools for solving a marginal problem
  with applications in bell non-locality and causal modeling}.
\newblock \emph{\bibinfo{journal}{Journal of Physics A: Mathematical and
  Theoretical}} \textbf{\bibinfo{volume}{51}}, \bibinfo{pages}{484002}
  (\bibinfo{year}{2018}).
\newblock \urlprefix\url{https://doi.org/10.1088%2F1751-8121%2Faae754}.

\bibitem{Hu2006}
\bibinfo{author}{Hu, X.} \& \bibinfo{author}{Ye, Z.}
\newblock \bibinfo{title}{Generalized quantum entropy}.
\newblock \emph{\bibinfo{journal}{Journal of Mathematical Physics}}
  \textbf{\bibinfo{volume}{47}}, \bibinfo{pages}{023502}
  (\bibinfo{year}{2006}).
\newblock \bibinfo{note}{\url{http://aip.scitation.org/doi/10.1063/1.2165794}}.

\bibitem{Audenaert2007}
\bibinfo{author}{Audenaert, K.}
\newblock \bibinfo{title}{Subadditivity of q-entropies for $q>1$}.
\newblock \emph{\bibinfo{journal}{Journal of Mathematical Physics}}
  \textbf{\bibinfo{volume}{48}}, \bibinfo{pages}{083507}
  (\bibinfo{year}{2007}).
\newblock
  \bibinfo{note}{\url{https://aip.scitation.org/doi/10.1063/1.2771542}}.

\bibitem{Leifer2013}
\bibinfo{author}{Leifer, M.~S.} \& \bibinfo{author}{Spekkens, R.~W.}
\newblock \bibinfo{title}{Towards a formulation of quantum theory as a causally
  neutral theory of {B}ayesian inference}.
\newblock \emph{\bibinfo{journal}{Phys. Rev. A}} \textbf{\bibinfo{volume}{88}},
  \bibinfo{pages}{052130} (\bibinfo{year}{2013}).
\newblock
  \bibinfo{note}{\url{https://link.aps.org/doi/10.1103/PhysRevA.88.052130}}.

\bibitem{Costa2016}
\bibinfo{author}{Costa, F.} \& \bibinfo{author}{Shrapnel, S.}
\newblock \bibinfo{title}{Quantum causal modelling}.
\newblock \emph{\bibinfo{journal}{New Journal of Physics}}
  \textbf{\bibinfo{volume}{18}}, \bibinfo{pages}{63032} (\bibinfo{year}{2016}).
\newblock
  \bibinfo{note}{\url{https://iopscience.iop.org/article/10.1088/1367-2630/18/6/063032}}.

\bibitem{Allen2017}
\bibinfo{author}{Allen, J.-M.~A.}, \bibinfo{author}{Barrett, J.},
  \bibinfo{author}{Horsman, D.~C.}, \bibinfo{author}{Lee, C.~M.} \&
  \bibinfo{author}{Spekkens, R.~W.}
\newblock \bibinfo{title}{Quantum common causes and quantum causal models}.
\newblock \emph{\bibinfo{journal}{Physical Review X}}
  \textbf{\bibinfo{volume}{7}}, \bibinfo{pages}{031021} (\bibinfo{year}{2017}).
\newblock
  \bibinfo{note}{\url{http://link.aps.org/doi/10.1103/PhysRevX.7.031021}}.

\bibitem{Pienaar2019}
\bibinfo{author}{Pienaar, J.}
\newblock \bibinfo{title}{A time-reversible quantum causal model}
  (\bibinfo{year}{2019}).
\newblock \bibinfo{note}{\url{http://arxiv.org/abs/1902.00129}}.

\bibitem{Kim2016}
\bibinfo{author}{Kim, J.~S.}
\newblock \bibinfo{title}{Tsallis entropy and general polygamy of multi-party
  quantum entanglement in arbitrary dimensions}.
\newblock \emph{\bibinfo{journal}{Physical Review A}}
  \textbf{\bibinfo{volume}{94}} (\bibinfo{year}{2016}).
\newblock
  \bibinfo{note}{\url{https://journals.aps.org/pra/abstract/10.1103/PhysRevA.94.062338}}.

\bibitem{Horsman2016}
\bibinfo{author}{Horsman, D.}, \bibinfo{author}{Heunen, C.},
  \bibinfo{author}{Pusey, M.~F.}, \bibinfo{author}{Barrett, J.} \&
  \bibinfo{author}{Spekkens, R.~W.}
\newblock \bibinfo{title}{Can a quantum state over time resemble a quantum
  state at a single time?}
\newblock \emph{\bibinfo{journal}{Proceedings of the Royal Society A}}
  \textbf{\bibinfo{volume}{473}}, \bibinfo{pages}{20170395}
  (\bibinfo{year}{2017}).
\newblock \bibinfo{note}{\url{http://dx.doi.org/10.1098/rspa.2017.0395}}.

\bibitem{Jamiolkowski1972}
\bibinfo{author}{Jamiołkowski, A.}
\newblock \bibinfo{title}{Linear transformations which preserve trace and
  positive semidefiniteness of operators}.
\newblock \emph{\bibinfo{journal}{Reports on Mathematical Physics}}
  \textbf{\bibinfo{volume}{3}}, \bibinfo{pages}{275 -- 278}
  (\bibinfo{year}{1972}).
\newblock
  \bibinfo{note}{\url{http://www.sciencedirect.com/science/article/pii/0034487772900110}}.

\bibitem{Choi1975}
\bibinfo{author}{Choi, M.-D.}
\newblock \bibinfo{title}{Completely positive linear maps on complex matrices}.
\newblock \emph{\bibinfo{journal}{Linear Algebra and its Applications}}
  \textbf{\bibinfo{volume}{10}}, \bibinfo{pages}{285 -- 290}
  (\bibinfo{year}{1975}).
\newblock
  \bibinfo{note}{\url{http://www.sciencedirect.com/science/article/pii/0024379575900750}}.

\bibitem{Petz2014}
\bibinfo{author}{Petz, D.} \& \bibinfo{author}{Virosztek, D.}
\newblock \bibinfo{title}{Some inequalities for quantum {T}sallis entropy
  related to the strong subadditivity}.
\newblock \emph{\bibinfo{journal}{Mathematical Inequalities and Applications}}
  \textbf{\bibinfo{volume}{18}} (\bibinfo{year}{2014}).
\newblock \bibinfo{note}{\url{http://mia.ele-math.com/18-41}}.

\end{thebibliography}


\appendix


\section{Quantum generalisations of Theorems \ref{theorem: mibound} and \ref{theorem: causal}}\label{appendix: quantum}

In the following, for a (finite dimensional) Hilbert space
$\mathcal{H}$, we use $\mathcal{L}(\mathcal{H})$ to represent the set
of linear operators on $\mathcal{H}$, $\mathcal{P}(\mathcal{H})$ to
represent the set of positive (semi-definite) operators on
$\mathcal{H}$, and $\mathcal{S}(\mathcal{H})$ to denote the set of
density operators on $\mathcal{H}$ (positive and trace 1).

Tsallis entropies as defined for classical random variables in
Section~\ref{sec: tsallis} are easily generalised to the quantum case
by replacing the probability distribution by a density matrix
\cite{Hu2006}. For a quantum system described by the density matrix
$\rho \in \mathcal{S}(\mathcal{H})$ on the Hilbert space $\mathcal{H}$
and $q>0$, the quantum Tsallis entropy is defined by
\begin{equation}
\label{eq: qtsallis1}
 S_q(\rho)=
 \begin{cases}
   -\Tr \rho^q\ln_q \rho, & q\neq 1.\\
    H(\rho), & q=1.
  \end{cases}
\end{equation}
where $H(\rho)= -\Tr \rho \ln \rho$ is the von-Neumann entropy of
$\rho$ and $\ln_q (x)=\frac{x^{1-q}-1}{1-q}$ as in Section~\ref{sec:
  tsallis}.\footnote{Analogously to the classical case we keep it
  implicit that if $\rho$ has any 0 eigenvalues these do not
  contribute to the trace.}

Given a density operator $\rho_{AB}\in\cS(\cH_A\ot\cH_B)$, the
conditional quantum Tsallis entropy of $A$ given $B$ can then be
defined by $S_q(A|B)_\rho=S_q(AB)-S_q(B)$, the mutual information
between $A$ and $B$ by $I_q(A:B)_\rho=S_q(A)+S_q(B)-S_q(AB)$, and for
$\rho_{ABC}\in\cS(\cH_A\ot\cH_B\ot\cH_C)$ the conditional Tsallis
information between $A$ and $B$ given $C$ is defined by
$I_q(A:B|C)_\rho=S_q(A|C)+S_q(B|C)-S_q(AB|C)$.  In this section we use
$d_S$ to represent the dimensions of the Hilbert space
$\mathcal{H}_S$.

The following properties of quantum Tsallis entropies will be useful for what follows.
\begin{enumerate}
    \item \textbf{Pseudo-additivity \cite{Tsallis1988}:} \label{prop:Qpseudo}If $\rho_{AB}=\rho_A\ot\rho_B$,
      then 
    \begin{equation}
        \label{eq:Qpseudoadd}
        S_q(AB)=S_q(A)+S_q(B)+(1-q)S_q(A)S_q(B)\,.
    \end{equation}
  \item \label{prop:Qupper} \textbf{Upper bound \cite{Audenaert2007}:} For all $q>0$, we have
    $S_q(A)\leq\ln_qd_A$ and equality is achieved if and only
    if $\rho_A=\id_A/d_A$.
    \item \textbf{Subadditivity \cite{Audenaert2007}:} For any density matrix $\rho_{AB}$ with marginals $\rho_A$ and $\rho_B$, the following holds for all $q\geq 1$,
    \begin{equation}
        \label{eq:Qsubadd}
        S_q(AB)\leq S_q(A)+S_q(B)\,.
    \end{equation}
    
\end{enumerate}

Using these we can generalize Theorem~\ref{theorem: mibound} to the
quantum case. This corresponds to the causal structure with two
independent quantum nodes and no edges in between them.

\begin{theorem}
\label{theorem: qmibound}
For all bipartite density operators in product form, i.e., $\rho_{AB}=\rho_A\otimes \rho_B$ with
$\rho_A\in\mathcal{S}(\mathcal{H}_A)$ and $\rho_B\in\mathcal{S}(\mathcal{H}_B)$, the quantum Tsallis mutual information $I_q(A:B)_{\rho}$ is upper bounded as follows for all $q>0$
\begin{equation*}
    I_q(A:B)_{\rho}\leq f(q,d_A,d_B)\,,
\end{equation*}
where the function $f(q,d_A,d_B)$ is given by
\begin{equation*}
    f(q,d_A,d_B)=\frac{1}{(q-1)}\left(1-\frac{1}{d_A^{q-1}}\right)\left(1-\frac{1}{d_B^{q-1}}\right)=(q-1)\ln_qd_A\ln_qd_B\,.
\end{equation*}
The bound is saturated if and only if
$\rho_{AB}=\frac{\mathds{1}_A}{d_A}\otimes\frac{\mathds{1}_B}{d_B}$.
\end{theorem}
\begin{proof}
  The proof goes through in the same way as the proof of
  Theorem~\ref{theorem: mibound} for the classical case
  (Properties~\ref{prop:Qpseudo} and~\ref{prop:Qupper} are analogous
  to those needed in the classical proof).
\end{proof}

Next, we generalise Theorem~\ref{theorem: causal} and
Corollaries~\ref{corollary: causal} and~\ref{corollary:
  mainconstraint}. This would correspond to the causal constraints on
quantum Tsallis entropies implied by the common cause causal structure
with $C$ being a complete common cause of $A$ and $B$ (which share no
causal relations among themselves). Here, one must be careful in
precisely defining the conditional mutual information and interpreting
it physically. For example, if the common case $C$ were quantum and
the nodes $A$ and $B$ were classical outcomes of measurements on $C$,
then $A$, $B$ and $C$ do not coexist and there is no joint state
$\rho_{ABC}$ in such a case. This is a significant difference in
quantum causal modelling compared to the classical case, and there
have been several proposals for how do deal with it \cite{Leifer2013,
  Costa2016, Allen2017, Pienaar2019}. In the following we
consider two cases:
\begin{enumerate}
\item When $C$ is classical, all 3 systems coexist and $\rho_{ABC}$
  can be described by a classical-quantum state (See
  Theorem~\ref{theorem: qcausal1}).
\item When $C$ is quantum, one approach is to view $\rho_{ABC}$ not as
  the joint state of the 3 systems but as being related to the
  Choi-Jamiolkowski representations of the quantum channels from $C$
  to $A$ and $B$ (See Section~\ref{ssec: noncoexisting}) as done in
  \cite{Allen2017}.
\end{enumerate}

The following Lemma proven in~\cite{Kim2016} is required for our
generalization of Theorem~\ref{theorem: causal} in the first case.
\begin{lemma}[\cite{Kim2016}, Lemma~1]
\label{lemma: condtsal}
Let $\mathcal{H}_A$ and $\mathcal{H}_Z$ be two Hilbert spaces and $\{\ket{z}\}_z$ be an orthonormal basis of $\mathcal{H}_Z$. Let $\rho_{AZ}$ be classical on $\mathcal{H}_Z$ with respect to this basis i.e., 
$$\rho_{AZ}=\sum_z p_z \rho_A^{(z)} \otimes \ket{z}\bra{z},$$
where $\sum_zp_z=1$ and $\rho_A^{(z)} \in S(\mathcal{H}_A)\ \forall z$. Then for all $q>0$, $$S_q(AZ)_{\rho}=\sum_zp_z^qS_q(\rho_A^{(z)})+S_q(Z),$$
where $S_q(Z)$ is the classical Tsallis entropy of the variable $Z$ distributed according to $P_Z$.
\end{lemma}
Note that the above Lemma immediately implies that
\begin{equation}\label{eq:cond}
  S_q(A|Z)_{\rho}=\sum_zp_z^qS_q(\rho_A^{(z)})\,.
\end{equation}

\begin{theorem}
\label{theorem: qcausal1}
Let $\rho_{ABC}=\sum_c p_c\rho_{AB}^{(c)}\otimes \proj{c}$, where $\rho_{AB}^{(c)}=\rho_A^{(c)}\otimes \rho_B^{(c)}$ $\forall c$, then, for all $q\geq1$,
$$I_q(A:B|C)_{\rho_{ABC}}\leq f(q,d_A,d_B)\,.$$

For $q>1$ the bound is saturated if and only if $\rho_{ABC}=\frac{\mathds{1}_A}{d_A}\otimes \frac{\mathds{1}_B}{d_B}\otimes \ket{c}\bra{c}_C$.
\end{theorem}
\begin{proof}
Using~\eqref{eq:cond} we have,
\begin{align*}
      I_q(A:B|C)_{\rho_{ABC}}
    &=S_q(A|C)_{\rho}+S_q(B|C)_{\rho}-S_q(AB|C)_{\rho}\\
    &=\sum_cp_c^q[S_q(\rho_A^{(c)})+S_q(\rho_B^{(c)})-S_q(\rho_{AB}^{(c)})]\\
    &=\sum_cp_c^qI_q(A:B)_{\rho_{AB}^{(c)}}\,.
\end{align*}
The rest of the proof is analogous to Theorem~\ref{theorem: causal},
where using the above, Theorem~\ref{theorem: qmibound} and defining
the set
$\mathcal{R}= \{\rho_{ABC} \in
\mathcal{H}_A\otimes\mathcal{H}_B\otimes\mathcal{H}_C:
\rho_{ABC}=\sum_cp_c\rho_A^{(c)}\otimes\rho_B^{(c)}\otimes\ket{c}\bra{c}\}$
we have,

\begin{align*}
\max\limits_{\mathcal{R}} I_q(A:B|C)_{\rho} &=\max\limits_{\mathcal{R}}\sum_c p_c^qI_q(A:B)_{\rho_{AB}^{(c)}}\\
&\leq \max\limits_{\{p_c\}_c}\sum_c p_c^q(c) \max\limits_{\{\rho_A^{(c)}\}_c,\{\rho_B^{(c)}\}_c}I_q(A:B)_{\rho_{AB}^{(c)}}\\
&=f(q,d_A,d_B)\,,
\end{align*}
where the last step follows because for all $q\geq1$, $\sum_cp_c^q$ is
maximized by deterministic distributions over $C$ with a maximum value
of $1$\footnote{For $q>1$ such deterministic distributions are the
  only way to obtain the bound.} and $I_q(A:B)_{\rho_{AB}^{(c)}}$ for
product states is maximised by the maximally mixed state over $A$ and
$B$ for all $c$ (Theorem~\ref{theorem: qmibound}). Thus, for $q>1$,
the bound is saturated if and only if
$\rho_{ABC}=\frac{\mathds{1}_A}{d_A}\otimes
\frac{\mathds{1}_B}{d_B}\otimes \ket{c}\bra{c}_C$ for some value $c$
of $C$.
\end{proof}

\subsection{A generalisation: when systems do not coexist}
\label{ssec: noncoexisting}
There is a fundamental problem with naively generalising classical conditional independences such as $p_{XY|Z}=p_{X|Z}p_{Y|Z}$ to the quantum case by replacing joint distributions by density matrices: it is not clear what is meant by a conditional quantum state e.g., $\rho_{A|C}$ since it is not clear what it means to condition on a quantum system, specially when the (joint state of the) system under consideration and the one being conditioned upon do not coexist. There are a number of approaches for tackling this problem, from describing quantum states in space and time on an equal footing \cite{Horsman2016} to quantum analogues of Bayesian inference \cite{Leifer2013} and causal modelling \cite{Costa2016, Allen2017, Pienaar2019}. In the following, we will focus on one such approach that is motivated by the framework of \cite{Allen2017}. Central to this approach is the Choi-Jamio\l{}kowski isomorphism \cite{Jamiolkowski1972, Choi1975} from which one can define conditional quantum states.

\begin{definition}[Choi state]
Let $\ket{\gamma}=\sum_i\ket{i}_R\ket{i}_{R^*}\in \mathcal{H}_R\otimes\mathcal{H}_{R^*}$, where $\mathcal{H}_{R^*}$ is the dual space to $\mathcal{H}_R$ and $\{\ket{i}_R\}_i$, $\{\ket{i}_{R^*}\}_i$ are orthonormal bases of $\mathcal{H}_R$ and $\mathcal{H}_{R^*}$ respectively. Given a channel $\mathcal{E}_{R|S}:\mathcal{S}(\mathcal{H}_R)\to\mathcal{S}(\mathcal{H}_S)$, the \emph{Choi state of
  the channel} is defined by
$$\rho_{S|R}=(\mathcal{E}_{R|S}\otimes\mathcal{I})(\ket{\gamma}\bra{\gamma})=\sum_{ij}\mathcal{E}(\ket{i}\bra{j}_R)\otimes\ket{i}\bra{j}_{R^*}\,.$$
Thus, $\rho_{S|R}\in\mathcal{P}(\mathcal{H}_S\otimes\mathcal{H}_{R^*})$.
\end{definition}

Now, if a quantum system $C$ evolves through a unitary channel
$\mathcal{E}_I(\cdot)=U'(\cdot)U'^{\dagger}$ to two systems $A'$ and
$B'$ where
$U':\mathcal{H}_C\rightarrow \mathcal{H}_{A'}\otimes\mathcal{H}_{B'}$,
it is reasonable to call the system $C$ a quantum common cause of the
systems $A'$ and $B'$. Further, this would still be reasonable if one
were to then perform local completely positive trace preserving (CPTP)
maps on the $A'$ and $B'$ systems. By the Stinespring dilation
theorem, these local CPTP maps can be seen as local isometries
followed by partial traces, and the local isometries can be seen as
the introduction of an ancilla in a pure state followed by a joint
unitary on the system and ancilla. This is illustrated in
Figure~\ref{fig: qchannel} and is compatible with the definition of
quantum common causes presented in \cite{Allen2017}.  In other words,
a system $C$ can be said to be a complete (quantum) common cause of
systems $A$ and $B$ if the corresponding
channel
$\mathcal{E}:\mathcal{S}(\mathcal{H}_C)\rightarrow
\mathcal{S}(\mathcal{H}_A\otimes\mathcal{H}_B)$ can be decomposed as
in Figure~\ref{fig: qchannel} for some choice of unitaries $U'$,
$U_A$, $U_B$ and pure states $\ket{\phi}_{E_A}$,
$\ket{\psi}_{E_B}$. Note that a more general set of channels fit the
definition of quantum common cause in Ref.~\cite{Allen2017} than we
use here; whether the theorems here extend to this case we leave as an open
question.

\begin{figure}[t]
    \centering
\includegraphics[scale=1.0]{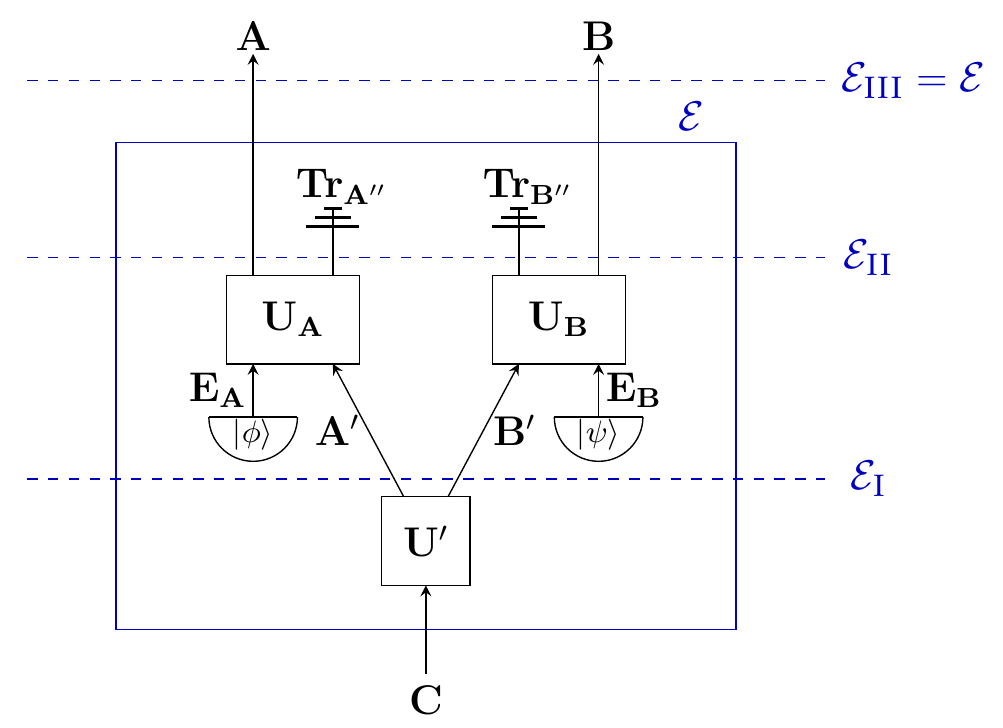}
    \caption{\textbf{A circuit decomposition of the channel $\mathbf{\mathcal{E}:\mathcal{S}(\mathcal{H}_C)\rightarrow \mathcal{S}(\mathcal{H}_A\otimes\mathcal{H}_B)}$ when $\mathbf{C}$ is a \emph{complete common cause} of $\mathbf{A}$ and $\mathbf{B}$:} If the map $\mathcal{E}$ from the system $C$ to the systems $A$ and $B$ can be decomposed as shown here, then $C$ is a complete common cause of $A$ and $B$ (\cite{Allen2017}). We build up our result step by step considering the channels given by $\mathcal{E}_{\textsc{i}}$ (unitary), $\mathcal{E}_{\textsc{ii}}$ (unitary followed by local isometries) and $\mathcal{E}_{\textsc{iii}}=\mathcal{E}$.}
    \label{fig: qchannel}
\end{figure}

In \cite{Allen2017} it is shown that whenever a system $C$ is a
complete common cause of systems $A$ and $B$ then the Shannon conditional mutual information evaluated
on the state $\tau_{ABC^*}=\frac{1}{d_A}\rho_{AB|C}$ satisfies
$I(A:B|C^*)_\tau=0$ where $\rho_{AB|C}$ is the Choi state of the channel
from $C$ to $A$ and $B$. We generalise this result to Tsallis
entropies for $q\geq 1$ for certain types of channels. We present the
result in three cases, each with increasing levels of generality. These are explained in Figure~\ref{fig: qchannel} and
correspond to the cases where the map from the complete common cause
$C$ to its children $A$ and $B$ is ({\sc i}) unitary ($\mathcal{E}_{\textsc{i}}=U'$); ({\sc ii}) unitary followed by local isometries ($\mathcal{E}_{\textsc{ii}}$); ({\sc iii}) Unitary followed by local isometries followed by partial traces on local systems ($\mathcal{E}_{\textsc{iii}}=\mathcal{E}$). 

\begin{lemma}\label{lemma: qstageI}
Let $\mathcal{E}_{\textsc{i}}: \mathcal{S}(\mathcal{H}_C)\rightarrow \mathcal{S}(\mathcal{H}_{A'}\otimes\mathcal{H}_{B'})$ be a unitary quantum channel i.e., $$\mathcal{E}_{\textsc{i}}(\cdot)= U'(\cdot)U'^{\dagger},$$
where $U': \mathcal{H}_C\rightarrow \mathcal{H}_{A'}\otimes\mathcal{H}_{B'}$ is an arbitrary unitary operator. If $\rho_{A'B'|C}$ is the corresponding Choi state, then the Tsallis conditional mutual information evaluated on the state $\tau_{A'B'C^*}=\frac{1}{d_C}\rho_{A'B'|C}\in \mathcal{S}(\mathcal{H}_{A'}\otimes\mathcal{H}_{B'}\otimes\mathcal{H}_{C^*})$ satisfies
  $$I_q(A':B'|C^*)_\tau= f(q,d_{A'},d_{B'}) \qquad \forall q>0.$$
\end{lemma}
\begin{proof}
 The conditional mutual information $I_q(A':B'|C^*)_\tau$ can be written as
 \begin{equation}
 \label{eq: qCMI-I}
     I_q(A':B'|C^*)_\tau=\frac{1}{q-1}\big(\Tr_{A'B'C^*}\tau_{A'B'C^*}^q+\Tr_{C^*}\tau_{C^*}^q-\Tr_{A'C^*}\tau_{A'C^*}^q-\Tr_{B'C^*}\tau_{B'C^*}^q\big).
 \end{equation}
 We will now evaluate every term in the above expression for the case where the channel that maps the $C$ system to the $A'$ and $B'$ systems is unitary. In this case, $\tau_{A'B'C^*}$ is a pure state and can be written as $\tau_{A'B'C^*}=\ket{\tau}\bra{\tau}_{A'B'C^*}$ where
 \begin{equation}
 \label{eq: vectau-I}
     \ket{\tau}_{A'B'C^*}=\frac{1}{\sqrt{d_C}}\sum_i U'\ket{i}_C\otimes \ket{i}_{C^*}.
 \end{equation}
 This means that
 $\Tr_{A'B'C^*}\tau_{A'B'C^*}^q=\Tr_{A'B'C^*}\tau_{A'B'C^*}$ $\forall
 q>0$. Since $\tau_{A'B'C^*}$ is a valid quantum state, it must be a trace one operator and we have
 \begin{equation}
  \label{eq: qterm1-I}
     \Tr_{A'B'C^*}\tau_{A'B'C^*}^q=1 \quad \forall q>0
 \end{equation}
 Further, we have $\tau_{C^*}=\Tr_{A'B'}\tau_{A'B'C^*}=\frac{\mathds{1}_{C^*}}{d_C}$ and hence
 \begin{equation}
  \label{eq: qterm2-I}
     \Tr_{C^*}\tau_{C^*}^q=\frac{1}{d_C^{q-1}}=\frac{1}{d_{A'}^{q-1}d_{B'}^{q-1}}.
 \end{equation}
 The second step follows from the fact that $U': \mathcal{H}_C\rightarrow \mathcal{H}_{A'}\otimes\mathcal{H}_{B'}$ is unitary so $d_C=d_{A'}d_{B'}$. 
 
 Now, the marginals over $A'$ and $B'$ are
 $\tau_{A'}=\Tr_{B'C^*}\tau_{A'B'C^*}=\frac{\mathds{1}_{A'}}{d_{A'}}$
 and
 $\tau_{B'}=\Tr_{A'C^*}\tau_{A'B'C^*}=\frac{\mathds{1}_{B'}}{d_{B'}}$. By the Schmidt decomposition of $\tau_{A'B'C^*}$, the non-zero
 eigenvalues of $\tau_{A'}$ are the same as those of  $\tau_{B'C^*}$.
 Since the Tsallis entropy depends only on the non-zero eigenvalues,
 $S_q(A')=S_q(B'C^*)$ and hence
 \begin{equation}
  \label{eq: qterm3-I}
     \Tr_{B'C^*}\tau_{B'C^*}^q=d_{A'}\Bigg(\frac{1}{d_{A'}^q}\Bigg)=\frac{1}{d_{A'}^{q-1}}\,.
 \end{equation}
 By the same argument it follows that
  \begin{equation}
   \label{eq: qterm4-I}
     \Tr_{A'C^*}\tau_{A'C^*}^q=d_{B'}\Bigg(\frac{1}{d_{B'}^q}\Bigg)=\frac{1}{d_{B'}^{q-1}}\,.
 \end{equation}
 Combining Equations~\eqref{eq: qCMI-I}-\eqref{eq: qterm4-I}, we have
 \begin{align}
     I_q(A':B'|C^*)_\tau=\frac{1}{q-1}\Bigg(1+\frac{1}{d_{A'}^{q-1}d_{B'}^{q-1}}-\frac{1}{d_{A'}^{q-1}}-\frac{1}{d_{B'}^{q-1}}\Bigg)
     =f(q,d_{A'},d_{B'}) \quad \forall q>0\,.
 \end{align}
\end{proof}

\begin{lemma}\label{lemma: qstageII}
Let $\mathcal{E}_{\textsc{ii}}: \mathcal{S}(\mathcal{H}_C)\rightarrow \mathcal{S}(\mathcal{H}_{\tilde{A}}\otimes\mathcal{H}_{\tilde{B}})$ be a quantum channel of the form $$\mathcal{E}_{\textsc{ii}}(\cdot)=(U_A\otimes U_B)\big[\ket{\phi}\bra{\phi}_{E_A}\otimes U'(\cdot)U'^{\dagger}\otimes \ket{\psi}\bra{\psi}_{E_B}\big](U_A\otimes U_B)^{\dagger},$$
where $U': \mathcal{H}_C\rightarrow \mathcal{H}_{A'}\otimes\mathcal{H}_{B'}$, $U_A:\mathcal{H}_{E_A}\otimes\mathcal{H}_{A'}\rightarrow \mathcal{H}_{\tilde{A}}$ and $U_B:\mathcal{H}_{B'}\otimes\mathcal{H}_{E_B}\rightarrow \mathcal{H}_{\tilde{B}}$ are arbitrary unitaries and $\ket{\phi}_{E_A}$ and $\ket{\psi}_{E_B}$ are arbitrary pure states. If $\rho_{\tilde{A}\tilde{B}|C}$ is the corresponding Choi state, then the Tsallis conditional mutual information evaluated on the state $\tau_{\tilde{A}\tilde{B}C^*}=\frac{1}{d_C}\rho_{\tilde{A}\tilde{B}|C}\in \mathcal{S}(\mathcal{H}_{\tilde{A}}\otimes\mathcal{H}_{\tilde{B}}\otimes\mathcal{H}_{C^*})$ satisfies
  $$I_q(\tilde{A}:\tilde{B}|C^*)_\tau= f(q,d_{A'},d_{B'}) \qquad \forall q>0.$$
\end{lemma}
\begin{proof}
  Note that the map $\mathcal{E}_{\textsc{ii}}$ is the unitary map
  $\mathcal{E}_{\textsc{i}}(\cdot)=U'(\cdot)U'^{\dagger}$ followed by
  local isometries $V_A$ and $V_B$ on the $A'$ and $B'$ systems
  respectively. Since the expression for the conditional mutual
  information $I_q(\tilde{A}:\tilde{B}|C^*)_\tau$ can be written in
  terms of entropies, which are functions of the eigenvalues of the
  relevant reduced density operators, and since the eigenvalues are
  unchanged by local isometries, this conditional mutual information
  is invariant under local isometries.  The rest of the proof is
  identical to that of Lemma~\ref{lemma: qstageI} resulting in
  \begin{equation}
      I_q(\tilde{A}:\tilde{B}|C^*)_\tau=  I_q(A':B'|C^*)_\tau=f(q,d_{A'},d_{B'}) \qquad \forall q>0.
  \end{equation}
\end{proof}

For the last case where $\mathcal{E}_{\textsc{iii}}(\cdot)=\Tr_{A''B''}\Big[(U_A\otimes U_B)\big[\ket{\phi}\bra{\phi}_{E_A}\otimes U'(\cdot)U'^{\dagger}\otimes \ket{\psi}\bra{\psi}_{E_B}\big](U_A\otimes U_B)^{\dagger}\Big]$, one could intuitively argue that tracing out systems could not increase the mutual information and one would expect that
\begin{equation}
\label{eq: wish}
   I_q(AA'':BB''|C^*)_\tau\geq I_q(A:B|C^*)_\tau.
\end{equation}  
Since
$I_q(AA'':BB''|C^*)_\tau=I_q(A:B|C^*)_\tau+I_q(AA'':B''|BC^*)_\tau+I_q(A'':B|AC^*)_\tau$,
Equation~\eqref{eq: wish} would follow from strong subadditivity used
twice i.e., $I_q(AA'':B''|BC^*)_\tau\geq 0$ and
$I_q(A'':B|AC^*)_\tau\geq 0$. However, it is known that strong
subadditivity does not hold in general for Tsallis entropies for
$q>1$~\cite{Petz2014}. Ref.~\cite{Petz2014} also provides a
sufficiency condition for strong subadditivity to hold for Tsallis
entropies. In the following Lemma, we provide another, simple
sufficiency condition that also helps bound the Tsallis mutual information $I_q(AA'':B|C)_{\tau}$ (or $I_q(A:BB''|C)_{\tau}$) corresponding to the map $\mathcal{E}_{\textsc{iii}}$ where only one of $A''$ or $B''$ is traced out but not both.

\begin{lemma}[Sufficiency condition for strong subadditivity of Tsallis entropies]
\label{lemma: sufficient}
If $\rho_{ABC}$ is a pure quantum state, then for all $q\geq 1$ we
have $I_q(A:B|C)_\rho\geq 0$.
\end{lemma}
\begin{proof}
  We have
  \begin{align*}
          I_q(A:B|C)=S_q(AC)+S_q(BC)-S_q(ABC)-S_q(C).
  \end{align*}
  Since $\rho_{ABC}$ is pure we have $S_q(ABC)=0$ $\forall q> 0$
  and (from the Schmidt decomposition argument mentioned earlier)
  $S_q(AC)=S_q(B)$, $S_q(BC)=S_q(A)$ and $S_q(C)=S_q(AB)$. Thus,
   \begin{align*}
          I_q(A:B|C)=S_q(A)+S_q(B)-S_q(AB)=I_q(A:B)\geq 0,
  \end{align*}
  which follows from subadditivity of quantum Tsallis entropies for $q\geq 1$ \cite{Audenaert2007}. In other words, for pure $\rho_{ABC}$, strong subadditivity of Tsallis entropies is equivalent to their subadditivity which holds whenever $q\geq 1$.
\end{proof}

\begin{corollary}
\label{corollary: qstageIII}
Let $\mathcal{E}^1_{\textsc{iii}}: \mathcal{S}(\mathcal{H}_C)\rightarrow \mathcal{S}(\mathcal{H}_{\tilde{A}}\otimes\mathcal{H}_{B})$ be a quantum channel of the form $$\mathcal{E}^1_{\textsc{iii}}(\cdot)=\Tr_{B''}\Big[(U_A\otimes U_B)\big[\ket{\phi}\bra{\phi}_{E_A}\otimes U'(\cdot)U'^{\dagger}\otimes \ket{\psi}\bra{\psi}_{E_B}\big](U_A\otimes U_B)^{\dagger}\Big],$$
where $U': \mathcal{H}_C\rightarrow \mathcal{H}_{A'}\otimes\mathcal{H}_{B'}$, $U_A:\mathcal{H}_{E_A}\otimes\mathcal{H}_{A'}\rightarrow\mathcal{H}_{\tilde{A}}\cong\mathcal{H}_{A}\otimes\mathcal{H}_{A''}$ and $U_B:\mathcal{H}_{B'}\otimes\mathcal{H}_{E_B}\rightarrow\mathcal{H}_{\tilde{B}}\cong\mathcal{H}_{B}\otimes\mathcal{H}_{B''}$ are arbitrary unitaries and $\ket{\phi}_{E_A}$ and $\ket{\psi}_{E_B}$ are arbitrary pure states. If $\rho_{\tilde{A}B|C}$ is the corresponding Choi state, then the Tsallis conditional mutual information evaluated on the state $\tau_{\tilde{A}BC^*}=\frac{1}{d_C}\rho_{\tilde{A}B|C}\in \mathcal{S}(\mathcal{H}_{\tilde{A}}\otimes\mathcal{H}_{B}\otimes\mathcal{H}_{C^*})$ satisfies
  $$I_q(\tilde{A}:B|C^*):=I_q(AA'':B|C^*)_\tau\leq f(q,d_{A'},d_{B'}) \qquad \forall q\geq 1.$$
\end{corollary}
\begin{proof}
  Since $I_q(AA'':BB''|C^*)_{\tau}=I_q(AA'':B|C^*)_{\tau}+I_q(AA'':B''|BC^*)_{\tau}$, the purity of $\tau_{\tilde{A}\tilde{B}C^*}=\tau_{AA''BB''C^*}$ and Lemma~\ref{lemma: sufficient} imply that $$I_q(AA'':BB''|C^*)_{\tau}\geq I_q(AA'':B|C^*)_{\tau}, \forall q\geq 1,$$ or (equivalently) in more concise notation,
  $$I_q(\tilde{A}:\tilde{B}|C^*)_{\tau}\geq I_q(\tilde{A}:B|C^*)_{\tau} \quad \forall q\geq 1.$$ Finally, using Lemma~\ref{lemma: qstageII} we obtain the required result.
\end{proof}

Now, for Equation~\eqref{eq: wish} to hold, we do not necessarily need strong subadditivity. Even if $I_q(A'':B|AC)_{\tau}\geq 0$ does not hold, Equation~\eqref{eq: wish} would still hold if $I_q(AA'':B''|BC)_{\tau}+I_q(A'':B|AC)_{\tau}\geq 0$. This motivates the following conjecture.

\begin{conjecture}\label{conjecture: qstageIII}
Let $\mathcal{E}_{\textsc{iii}}: \mathcal{S}(\mathcal{H}_C)\rightarrow \mathcal{S}(\mathcal{H}_{A}\otimes\mathcal{H}_{B})$ be a quantum channel of the form $$\mathcal{E}_{\textsc{iii}}(\cdot)=\Tr_{A''B''}\Big[(U_A\otimes U_B)\big[\ket{\phi}\bra{\phi}_{E_A}\otimes U'(\cdot)U'^{\dagger}\otimes \ket{\psi}\bra{\psi}_{E_B}\big](U_A\otimes U_B)^{\dagger}\Big],$$
where $U': \mathcal{H}_C\rightarrow \mathcal{H}_{A'}\otimes\mathcal{H}_{B'}$, $U_A:\mathcal{H}_{E_A}\otimes\mathcal{H}_{A'}\rightarrow \mathcal{H}_{A}\otimes\mathcal{H}_{A''}$ and $U_B:\mathcal{H}_{B'}\otimes\mathcal{H}_{E_B}\rightarrow \mathcal{H}_{B}\otimes\mathcal{H}_{B''}$ are arbitrary unitaries and $\ket{\phi}_{E_A}$ and $\ket{\psi}_{E_B}$ are arbitrary pure states. If $\rho_{AB|C}$ is the corresponding Choi state, then the Tsallis conditional mutual information evaluated on the state $\tau_{ABC^*}=\frac{1}{d_C}\rho_{AB|C}\in \mathcal{S}(\mathcal{H}_{A}\otimes\mathcal{H}_{B}\otimes\mathcal{H}_{C^*})$ satisfies
  $$I_q(A:B|C^*)_\tau\leq f(q,d_{A'},d_{B'}) \qquad \forall q\geq 1.$$
\end{conjecture}

Notice that in Corollary~\ref{corollary: qstageIII} and
Conjecture~\ref{conjecture: qstageIII}, the bounds are functions of
$d_{A'}$ and $d_{B'}$ and not of the dimensions of the systems $A$ and
$B$ (those in the quantity on the left hand side). In the case that
$d_A\geq d_{A'}$ and $d_B\geq d_{B'}$, the fact that $f(q,d_A,d_B)$ is
a strictly increasing function of $d_A$ and $d_B$ $\forall q\geq 0$
allows us to write $I_q(\tilde{A}:B|C^*)_\tau\leq
f(q,d_{\tilde{A}},d_B)$ and $I_q(A:B|C^*)_\tau\leq f(q,d_A,d_B)$ under
the conditions of Corollary~\ref{corollary: qstageIII} and
Conjecture~\ref{conjecture: qstageIII} respectively. However, if
$d_A\leq d_{A'}$ and/or $d_B\leq d_{B'}$, the bounds
$f(q,d_{\tilde{A}},d_B)$ and $f(q,d_A,d_B)$ are tighter than the bound
$f(q,d_{A'},d_{B'})$ and so not implied. However, based on the several examples that we have checked, we further conjecture the following.

\begin{conjecture}\label{conjecture: qstageIII2}
Under the same conditions as Conjecture~\ref{conjecture: qstageIII}
  $$I_q(A:B|C^*)_\tau\leq f(q,d_A,d_B) \qquad \forall q\geq 1.$$
\end{conjecture}

Further, it is shown in \cite{Allen2017} that if $C$ is a \emph{complete common cause} of $A$ and $B$ then the corresponding Choi state, $\rho_{AB|C}$ decomposes as $\rho_{AB|C}=(\rho_{A|C}\otimes\mathds{1}_B)(\mathds{1}_A\otimes\rho_{B|C})$ or $\rho_{AB|C}=\rho_{A|C}\rho_{B|C}$ in analogy with the classical case where if a classical random variable $Z$ is a common cause of the random variables $X$ and $Y$, then the joint distribution over these variables factorises as $p_{XY|Z}=p_{X|Z}p_{Y|Z}$. Then we have that $\tau_{ABC^*}=\frac{1}{d_C}\rho_{AB|C}=\frac{1}{d_C}\rho_{A|C}\rho_{B|C}$. By further analogy with the classical results of Section~\ref{sec: tsalcaus}, one may also consider
instead a state of the form
$\hat{\sigma}_{ABCC^*}=\sigma_C\otimes\frac{1}{d_C}\rho_{A|C}\rho_{B|C}=\sigma_C\otimes\tau_{ABC^*}$,
where $\sigma_C\in\mathcal{S}(\mathcal{H}_C)$.\footnote{This is the
  analogue of the statement $p_{ABC}=p_Cp_{A|C}p_{B|C}$ for
  probability distributions.}  Note that $\hat{\sigma}_{ABCC^*}$ is a
valid density operator on $\mathcal{H}_A\otimes\mathcal{H}_B\otimes\mathcal{H}_C\otimes\mathcal{H}_{C^*}$.

\begin{lemma}\label{lemma: sigma}
  The state $\hat{\sigma}_{ABCC^*}=\sigma_C\otimes\tau_{ABC^*}$ defined above
  satisfies $$I_q(A:B|CC^*)_{\hat{\sigma}}\leq f(q,d_A,d_B)\,,$$
whenever $I_q(A:B|C^*)_{\tau}\leq f(q,d_A,d_B)$ holds for the state $\tau_{ABC^*}=\frac{1}{d_A}\rho_{AB|C}$, where $\rho_{AB|C}$ represents the quantum channel from $C$ to $A$ and $B$ and $\sigma_C$ is the input quantum state to this channel.
\end{lemma}
\begin{proof}
  Since $\hat{\sigma}$ is a product state between the $C$ and $ABC^*$ subsystems, by the pseudo-additivity of quantum Tsallis entropies and
  the chain rule we have
  \begin{align*}
    I_q(A:B|CC^*)_{\hat{\sigma}}&=S_q(ACC^*)+S_q(BCC^*)-S_q(ABCC^*)-S_q(CC^*)\\                   &=S_q(AC^*)+S_q(BC^*)-S_q(ABC^*)-S_q(C^*)\\
    &\ \ \ -(q-1)S_q(C)\left(S_q(AC^*)+S_q(BC^*)-S_q(ABC^*)-S_q(C^*)\right)\\
                                &=(1-(1-q)S_q(C))I(A:B|C^*)\\
                                &=\Tr(\sigma_C^q)I(A:B|C^*)\,.
  \end{align*}

Now let $p_c$ be the distribution whose entries are the eigenvalues
of $\sigma_C$.  We have $\Tr(\sigma_C^q)=\sum_cp_c^q$.  Thus if $q>1$,
$\sum_cp_c^q\leq 1$ with equality if and only if $p_c=1$ for
some value of $c$.  It follows that
$$I_q(A:B|CC^*)_{\hat{\sigma}}\leq I_q(A:B|C^*)_\tau\,.$$
Therefore, if $I_q(A:B|C^*)_\tau\leq f(q,d_A,d_B)$, we also have
$I_q(A:B|CC^*)_{\hat{\sigma}}\leq f(q,d_A,d_B)$.
\end{proof}
\end{document}